\documentclass[11pt]{amsart}
\usepackage[utf8]{inputenc}
\usepackage[T1]{fontenc}
\usepackage[english]{babel}
\usepackage{lmodern,amsmath,amssymb,amsthm,mathtools,microtype}
\usepackage[margin=28mm]{geometry}
\usepackage{enumitem}
\usepackage{xcolor}
\definecolor{linkblue}{RGB}{25,63,105}
\usepackage[colorlinks=true,linkcolor=linkblue,citecolor=linkblue,urlcolor=linkblue]{hyperref}
\setlist[enumerate]{itemsep=2pt,topsep=3pt}
\newtheorem{theorem}{Theorem}[section]
\newtheorem{proposition}[theorem]{Proposition}
\newtheorem{lemma}[theorem]{Lemma}
\newtheorem{corollary}[theorem]{Corollary}
\theoremstyle{remark}
\newtheorem{remark}[theorem]{Remark}
\numberwithin{equation}{section}

\newcommand{\E}{\mathbb E}
\newcommand{\Prob}{\mathbb P}
\newcommand{\R}{\mathbb R}
\newcommand{\N}{\mathbb N}
\newcommand{\Q}{\mathbb Q}
\newcommand{\ga}{\gamma}
\newcommand{\ind}{\mathbf 1}
\newcommand{\op}{\mathrm{op}}
\newcommand{\HS}{\mathrm{HS}}
\newcommand{\bul}{\bullet}
\DeclareMathOperator{\Var}{Var}
\DeclareMathOperator{\Ent}{Ent}
\DeclareMathOperator{\sgn}{sgn}
\newcommand{\norm}[1]{\lVert#1\rVert}

\begin{document}

\title[Thermodynamic limit of the Hopfield model]{Existence of the thermodynamic limit for the Hopfield model}

\author{Adriano Barra}
\address{Adriano Barra: Dipartimento di Scienze di Base e Applicazioni
all'Ingegneria, Sapienza Universit\`a di Roma, Italy;
GNFM, Istituto Nazionale d'Alta Matematica, Sezione di Roma, Italy;
INFN, Istituto Nazionale di Fisica Nucleare, Sezione di Roma, Italy;
CNR, Nanotec, Salento Unit, Lecce, Italy}
\email{adriano.barra@uniroma1.it}

\author{Fabrizio Durante}
\address{Fabrizio Durante: Dipartimento di Matematica e Fisica
``Ennio De Giorgi'', Universit\`a del Salento, Lecce, Italy}
\email{fabrizio.durante@unisalento.it}

\date{30 October 2026}

\begin{abstract}
We prove that the quenched statistical pressure of the Hopfield model
with $K\sim\alpha N$ independent symmetric Rademacher patterns has a
thermodynamic limit at every load $\alpha\in(0,\infty)$, every finite
inverse temperature and every uniform external field. No replica-symmetry,
order-parameter concentration or non-condensation hypothesis is required.
The limit is the supremum over $r$ of the limiting pressures of mixed
models retaining $r$ binary patterns in a Gaussian bulk. The Hubbard--Stratonovich representation permits extraction of the
finitely many patterns whose hidden coordinates exceed $\delta\sqrt N$.
A collective interpolation replaces the remaining binary patterns by
Gaussian ones with error $O(\delta^2)$ uniform in the volume. Convex
concentration on the discrete cube controls the disorder-dependent
selection of retained patterns. For the mixed models, we prove a Gaussian hidden-spin truncation bound
uniform in the volume and in an arbitrary Ising potential, and establish
convergence for Gaussian bulks with finitely many site-field types.
These estimates reduce the problem to the results of Ho and of Chen,
Ho and Issa on non-convex vector spin glasses with compact spins. For Gaussian patterns, the argument gives an alternative proof of the
existence result announced by Camilli, Contucci and Mingione and extends
it to a uniform field and an additional binary pattern. For both pattern laws, the pressure has exponentially small deviation
probabilities and converges almost surely. The internal energy and
entropy converge wherever the limiting pressure is differentiable.
An additional binary Mattis pattern admits a scalar variational reduction. At
zero field the Rademacher pressure dominates the Gaussian pressure with
one additional binary pattern, which in turn dominates the Gaussian
pressure. Note that these results establish existence of the asymptotic limit but do not give an explicit
 formula.
\newline
The proofs were generated through extensive interaction with the AI
models Claude (Opus~5.5) and ChatGPT (Sol~6.1), based on the problem
formulation and guidance of the authors.

\end{abstract}

\maketitle

\section{Introduction}\label{sec:intro}

The thermodynamic limit of a disordered system is taken as the system
size $N\to\infty$ at fixed intensive parameters. Establishing its
existence provides a basis for the analysis of limiting thermodynamic
observables. For mean-field spin glasses, including the
Sherrington--Kirkpatrick model~\cite{SK75}, Guerra and
Toninelli~\cite{GT02} proved convergence of the statistical pressure by
interpolation. Subsequent work established the corresponding replica-symmetry-breaking
predictions: the Parisi variational
formula~\cite{Par80} follows from Guerra's interpolation
bounds~\cite{Gue03} and Talagrand's matching result~\cite{Tal06}, while
Panchenko~\cite{Pan13} proved that the Ghirlanda--Guerra
identities~\cite{GG98} imply ultrametricity. These results address
distinct questions: existence of the limit, its value, and the limiting
overlap structure.

For associative neural networks, the rigorous theory remains less
complete. The Hopfield model~\cite{Hop82} is an Ising system whose
interaction matrix is a sum of rank-one pattern matrices; its
statistical mechanics at extensive load, where the number $K$ of stored
patterns grows proportionally to the number $N$ of neurons, was developed by Amit, Gutfreund and Sompolinsky~\cite{AGS85a,AGS85b,AGS87}
using the replica method at the replica symmetric level of description. The Hubbard--Stratonovich transformation
represents the model as a bipartite system of $N$ Ising spins and $K$
Gaussian ``hidden'' spins, also interpretable as a hybrid Boltzmann
machine~\cite{BGG10,BGG11,BBCS12}. Blocking yields a vector spin glass
whose covariance is the product of two partial traces of the overlap
matrix. This covariance is not convex, so neither the
Guerra--Toninelli argument nor the classical theory of multi-species and
vector spin glasses~\cite{Pan15,Pan18,BCMT15} applies directly.
Configurations may also correlate with individual patterns, producing
signal (Mattis) magnetizations~\cite{Mat76}.

Rigorous contributions to the extensively loaded Hopfield model include
retrieval Gibbs states~\cite{BGP94,BGP95}, self-averaging~\cite{Bov94},
concentration-based replica-symmetric equations~\cite{PST94}, and
results at high temperature or small load where the replica-symmetric
picture can be justified~\cite{PST94,BG98,Tal98,Tal11}. For Gaussian
patterns, the annealed region $\beta\le(1+\sqrt\alpha)^{-1}$ is
understood~\cite{BG08,BGGT12}. Agliari, Barra, Bianco, Fachechi and
Pallara~\cite{ABBFP} obtained a hard--soft interpolation identity and
proved convergence under replica-symmetric concentration assumptions.
Concentration around an expected statistical pressure, however, does
not by itself establish that the expectation converges, and convergence throughout the parameter space has remained open for
binary patterns; see the discussion
in~\cite{ABBFP}.

For Gaussian patterns, Camilli, Contucci and Mingione~\cite{CCM} have
announced convergence of the quenched statistical pressure at every
$\alpha,\beta>0$, without an external field or additional pattern,
together with an explicit variational principle for the limit. Our approach also treats a uniform external field and an additional
binary pattern, and extends to Rademacher patterns. It does not determine
the limit explicitly (see Remark~\ref{rem:CCM}).

We prove existence of the thermodynamic limit of the quenched
statistical pressure of the Hopfield model with independent
symmetric Rademacher patterns, at every load $\alpha\in(0,\infty)$, every
finite inverse temperature and every uniform external field, with a
limit that depends on $(K(N))$ only through $\alpha$
(Theorems~\ref{thm:main} and~\ref{thm:Rrep}). For binary patterns,
convergence was previously known only in restricted regions of the
$(\alpha,\beta)$ plane, typically at high temperature or small load
where the replica-symmetric picture can be justified~\cite{BG98,Tal98,
Tal11}, or conditionally on hypotheses on the Gibbs measure, such as
self-averaging of the Edwards--Anderson order parameter~\cite{PST94} or
replica-symmetric concentration of the overlaps~\cite{ABBFP}.

The proof applies to all $\alpha>0$, $\beta<\infty$ and uniform fields,
without replica symmetry or concentration of overlaps or Mattis
magnetizations. It also requires no non-condensation assumption, that
is, no hypothesis that typical configurations have vanishing overlaps
with the patterns. This distinction matters for disorder comparison: the direct replacement of all binary patterns
by Gaussian ones (Lemma~\ref{lem:lower}) yields only the inequality
$A^R\ge A^G$, because its error is controlled only for bounded hidden
coordinates and a vanishing truncation error is available only on the Gaussian
side. No binary-to-Gaussian universality is assumed or proved; the two
limiting pressures are related by inequalities
(Corollary~\ref{cor:compare}).

The pressure also concentrates around its mean with exponentially small
deviation probabilities and converges almost surely
(Theorem~\ref{thm:concentration}). The almost-sure statement holds for
any coupling of the pattern matrices across volumes. Convexity in $\beta$ then transfers convergence from the
pressure to its derivative: the internal energy and the entropy converge
at every differentiability point of the limiting pressure, hence outside
an at most countable set of inverse temperatures, and one-sided bounds
hold at the remaining points (Theorem~\ref{thm:energy}). Letting
$\beta\to\infty$ gives convergence of the ground-state energy
(Corollary~\ref{cor:ground}).

The passage from Gaussian to binary disorder requires an additional
comparison. The compact vector-spin theorem of Ho
and of Chen, Ho and Issa~\cite{Ho26,CHI} (Theorem~\ref{thm:CHI})
requires the interaction to be a Gaussian process with a prescribed
covariance. Binary disorder is not of this form, so the theorem cannot
be applied to the binary model itself, and replacing all binary patterns by Gaussian ones gives only a
one-sided bound. For the upper bound, we retain the patterns whose
hidden coordinates may be large and replace the others. This \emph{finite-pattern extraction} reduces the binary model
to Gaussian bulks in which finitely many binary patterns are retained,
and to which the compact theorem applies through
Theorem~\ref{thm:types}.

The retained patterns are not fixed in advance. After a radial cutoff
$\norm z^2\le\varrho N$ of the hidden spins, we partition the hidden
integral into sectors indexed by the set $S$ of hidden coordinates
larger than $\delta\sqrt N$, with $|S|\le1+\lfloor\varrho/\delta^2
\rfloor$. On each sector the patterns in $S$, associated with
potentially macroscopic hidden coordinates, are retained as binary and
all others are replaced. The proof does not require identifying
condensed equilibrium states; Remark~\ref{rem:collective} explains the
relation between large hidden coordinates and Mattis overlaps.

The extraction uses three estimates. The first, a
\emph{collective replacement} (Lemma~\ref{lem:small}), replaces all
small-coordinate columns in a single interpolation under one common
Gibbs measure and sums the fourth moments before bounding them. The resulting error is $O(\beta^2\varrho\delta^2)$, uniformly
in $N$, $K$, $S$ and $h$, whereas an entry-by-entry replacement would
diverge (Remark~\ref{rem:collective}). The second estimate controls \emph{sector selection} (Section~\ref{sec:selection}). The dominant
sector depends on the disorder, but the number of sectors is polynomial
in $N$ and the sector pressures are convex in the disorder with
Lipschitz constants $O(N^{-1/2})$; convex concentration on the discrete
cube bounds the cost of selecting the dominant sector inside the quenched
expectation by $O(\sqrt{\log N/N})$. The third estimate is a \emph{lower comparison} (Lemma~\ref{lem:lower}), in which an arbitrary
finite set of binary patterns is retained and the rest are replaced
after a truncation of the hidden spins at level $N^{1/6}$, the Gaussian
side being controlled by Proposition~\ref{prop:cut}. Taking
$N\to\infty$ at fixed $\delta$, and then $\delta\downarrow0$, gives
$A^R=\sup_rp_r$ without any uniformity in $r$. Replacement
methods~\cite{Cha06}, spin-glass universality~\cite{CH06}, bipartite
comparison~\cite{Gen12} and mixed Boolean--Gaussian models~\cite{ABLT}
provide background, but no unrestricted universality statement is
invoked.

Both comparisons rely on the Gaussian analysis, which also proves
convergence for Gaussian patterns. At
$h=0$, Theorem~\ref{thm:main} for Gaussian patterns provides an alternative proof of the existence
result announced in~\cite{CCM}. Its first ingredient is a \emph{uniform
truncation} (Proposition~\ref{prop:cut}). The hidden spins are
unbounded, while Theorem~\ref{thm:CHI} requires compact spins;
truncating them at level $L$ changes the pressure by at most
$C_{\rm cut}L^{-2}$, uniformly in the volume, in a compact load interval
and in an \emph{arbitrary} Ising potential. The bound follows from the
rotational invariance of the Gaussian pattern matrix and requires no
assumption on the Gibbs measure. It allows the limits $N\to\infty$ and
$L\to\infty$ to be exchanged, and its uniformity in the potential permits comparisons with arbitrary
extensive potentials in
Sections~\ref{sec:mixed}--\ref{sec:sectors}.

The second ingredient extends the Gaussian result to bulks whose visible
sites carry \emph{finitely many types} of fields, with arbitrary
proportions, uniformly on compact sets of fields
(Theorem~\ref{thm:types}). The mixed models require this extension,
since their types are the sign patterns of the retained binary columns,
and a scalar uniform-field result would not suffice. The same analysis treats a uniform field and an additional binary
Mattis pattern (Theorem~\ref{thm:Mattis}), with explicit concentration bounds
(Proposition~\ref{prop:explicit-conc}).

The representation $A^R=\sup_rp_r$ leaves the optimal number of
retained patterns undetermined, including whether one suffices. By
Corollary~\ref{cor:compare}, one suffices at zero field exactly when the
binary pressure equals the Gaussian pressure with one additional binary
pattern. The number of retained patterns is not a replica-symmetry-breaking
depth. The argument does not prove that the remainder in the hard--soft
interpolation of~\cite{ABBFP} vanishes.
Broken-replica interpolation methods~\cite{Gue03,AABO} provide related
background, but no identification with their variational functionals is
needed below.

The proof begins with the external convergence result in
Theorem~\ref{thm:CHI}. Together with
Proposition~\ref{prop:cut}, it yields Theorem~\ref{thm:types} and hence
Theorem~\ref{thm:main} for Gaussian patterns
(Section~\ref{sec:gauss}). Theorem~\ref{thm:types} then gives the mixed-model limits (Proposition~\ref{prop:mixedlimit}), which,
combined with Lemmas~\ref{lem:lower} and~\ref{lem:small} and the sector
selection, give Theorem~\ref{thm:Rrep} and Theorem~\ref{thm:main} for
binary patterns (Sections~\ref{sec:mixed}--\ref{sec:closing}).
Concentration yields sample convergence
(Section~\ref{sec:concentration}). Convexity gives the results on energy,
entropy and ground states (Section~\ref{sec:energy}), and gauge
invariance yields the reduction for an additional pattern
(Section~\ref{sec:Mattis}). Section~\ref{sec:model} introduces the
models and states the main results, Section~\ref{sec:prelim} collects
the preliminary tools, and Section~\ref{sec:discussion} discusses the
results.

\section{Models and main results}\label{sec:model}

We write $\Sigma_N:=\{-1,1\}^N$, $m(\sigma):=N^{-1}\sum_i\sigma_i$, and
$\ga$ for the standard Gaussian measure on $\R$; $Y$ denotes a standard
Gaussian random variable or vector. For a matrix $M$, $\norm M_\op$ and
$\norm M_\HS$ are the operator and Hilbert--Schmidt norms. The
expectation $\E$ is over all disorder variables present; subscripts
indicate a restricted expectation. The symbol $\omega$ denotes the Gibbs
expectation associated with the partition function under consideration.
Constants $C$ may change between occurrences; their parameter dependence
is indicated when relevant. All logarithms are natural.

For a real $N\times K$ matrix $J$ with columns $J^\mu\in\R^N$ let
\begin{equation}\label{eq:Ham}
H_N(\sigma;J):=-\frac1{2N}\sum_{\mu=1}^K(J^\mu\cdot\sigma)^2
=-\frac1{2N}\norm{J^T\sigma}^2,\qquad \sigma\in\Sigma_N .
\end{equation}
The diagonal terms $i=j$ are retained; see Remark~\ref{rem:diag}. For
$\beta\ge0$ and a potential $V:\Sigma_N\to\R$ put
\begin{equation}\label{eq:fV}
f^V_N(J):=\frac1N\log\sum_{\sigma\in\Sigma_N}
e^{-\beta H_N(\sigma;J)+V(\sigma)},\qquad
c_V:=\log2+\frac1N\max_{\Sigma_N}|V| ,
\end{equation}
suppressing the dependence on $\beta$.

We consider two pattern laws. Let $G=(g_i^\mu)$ be an $N\times K$ matrix
of independent standard Gaussian variables and $\varepsilon=
(\varepsilon_i^\mu)$ an $N\times K$ matrix of independent symmetric
Rademacher variables. Throughout, $\bul\in\{G,R\}$ labels the pattern
law, with $J^G:=G$ and $J^R:=\varepsilon$. For $h\in\R$ let
\begin{equation}\label{eq:Z}
Z^\bul_{N,K}(\beta,h):=\sum_{\sigma\in\Sigma_N}
e^{-\beta H_N(\sigma;J^\bul)+hNm(\sigma)},\qquad
\widehat A^\bul_{N,K}:=\frac1N\log Z^\bul_{N,K},\qquad
A^\bul_{N,K}:=\E\widehat A^\bul_{N,K}.
\end{equation}
Thus $\widehat A^\bul_{N,K}(\beta,h)=f^{hNm}_N(J^\bul)$. Following the
convention of~\cite{GT02}, $\widehat A^\bul_{N,K}$ is the (random)
\emph{statistical pressure} and $A^\bul_{N,K}$ the \emph{quenched
statistical pressure}. For $\beta>0$ the free energy is 
$$
F^\bul_{N,K}:=-(\beta N)^{-1}\log Z^\bul_{N,K},
$$
so that $\widehat A^\bul_{N,K}=-\beta
F^\bul_{N,K}$; the statistical pressure is defined for all $\beta\ge0$,
and all results are stated for it. The
field $h$ is dimensionless and conjugate to $m$; at $h=0$ we recover
the ordinary Hopfield model. We write $Z^\bul_{N,K}(h)$,
$A^\bul_{N,K}(h)$ when $\beta$ is fixed and $Z^\bul_{N,K}(\beta)$,
$A^\bul_{N,K}(\beta)$ when $h=0$, and set $Z^\bul_{N,0}(\beta,h):=
(2\cosh h)^N$. Limits are denoted $A^G$ and $A^R$; their equality is
neither assumed nor proved.

For an integer $r\ge1$ and $K\ge r$, let $\zeta^1,\dots,\zeta^r\in
\Sigma_N$ be independent patterns with independent symmetric Rademacher
entries, and let $G'$ be an independent standard Gaussian $N\times(K-r)$
matrix. The mixed model has pattern matrix $(\zeta^1,\dots,\zeta^r,G')$:
\begin{equation}\label{eq:mixed}
Z^{(r)}_{N,K}(\beta,h):=\sum_\sigma\exp\Big\{\frac\beta{2N}
\norm{G'^T\sigma}^2+\frac\beta{2N}\sum_{a=1}^r(\zeta^a\cdot\sigma)^2
+hNm(\sigma)\Big\},\qquad
p^{(r)}_{N,K}:=\frac1N\E\log Z^{(r)}_{N,K}.
\end{equation}
These are auxiliary comparison models for binary patterns.
For fixed $r$, the sequence $p^{(r)}_{N,K(N)}$ is considered for
$N\ge N_r$, with $N_r$ chosen so that $K(N)\ge r$ for all $N\ge N_r$.
Such an $N_r$ exists when $K(N)/N\to\alpha>0$; the terms with $K(N)<r$
are left undefined.

Let $\eta\in\Sigma_N$ be random, with arbitrary law (which may depend
on $N$ and need not have independent coordinates), independent of
$J^\bul$. The model with the additional (Mattis) pattern $\eta$ is
\begin{equation}\label{eq:ZM}
H^M_N(\sigma):=H_N(\sigma;J^\bul)-\frac{(\eta\cdot\sigma)^2}{2N},\qquad
\widehat A^{\bul,M}_{N,K}(\beta):=\frac1N\log\sum_{\sigma\in\Sigma_N}
e^{-\beta H^M_N(\sigma)},\qquad A^{\bul,M}_{N,K}:=\E\widehat
A^{\bul,M}_{N,K}.
\end{equation}
There is no additional uniform field in this definition. Models of this
type, with symmetric binary $\eta$, are considered in~\cite{ABBFP,ABLT}.
We call $\eta$ an \emph{additional} pattern: no condensation, i.e.\ no
nonzero equilibrium overlap with $\eta$, is asserted below.

\subsection{Main results}

\begin{theorem}[Thermodynamic limit]\label{thm:main}
Let $\bul\in\{G,R\}$, $\alpha\in(0,\infty)$ and $K(N)/N\to\alpha$. For
every $\beta\ge0$ and $h\in\R$ the limit
\[
A^\bul(\alpha,\beta,h):=\lim_{N\to\infty}A^\bul_{N,K(N)}(\beta,h)
\]
exists, is finite and depends on $(K(N))$ only through $\alpha$. The
convergence is locally uniform in $(\beta,h)\in[0,\infty)\times\R$. We
write $A^\bul(\alpha,\beta):=A^\bul(\alpha,\beta,0)$.
\end{theorem}

\begin{theorem}[Binary patterns through finitely many retained patterns]
\label{thm:Rrep}
Let $\alpha\in(0,\infty)$ and $K(N)/N\to\alpha$. For every fixed
$r\ge1$, $K(N)\ge r$ for all sufficiently large $N$. For every
$\beta\ge0$ and $h\in\R$ the limit $p_r(\alpha,\beta,h):=\lim_N
p^{(r)}_{N,K(N)}(\beta,h)$, taken along these sufficiently large $N$,
exists, is finite and depends on $(K(N))$ only through $\alpha$. The
convergence is locally uniform on $[0,\infty)\times\R$, including compact
sets that meet $\beta=0$. The sequence $r\mapsto p_r(\alpha,\beta,h)$ is
nondecreasing and bounded, and
\begin{equation}\label{eq:main}
\quad A^R(\alpha,\beta,h)=\sup_{r\ge1}p_r(\alpha,\beta,h)
=\lim_{r\to\infty}p_r(\alpha,\beta,h).\quad
\end{equation}
\end{theorem}

The number $r$ of retained patterns is fixed while taking the
thermodynamic limit and may increase afterwards.
Formula~\eqref{eq:main} does not assert that $r=1$ suffices.

\begin{theorem}[Concentration and almost-sure limit]
\label{thm:concentration}$\,$
\begin{enumerate}[label=\textup{(\alph*)}]
\item \emph{Exponential concentration.} Let $\bul\in\{G,R\}$, $\beta>0$,
$\alpha_+>0$, $\bar c>0$ and $t>0$. There are $c>0$ and $C<\infty$ such
that for all $N,K\ge1$ with $K\le\alpha_+N$ and all deterministic
$V:\Sigma_N\to\R$ with $c_V\le\bar c$,
\[
\Prob\big(|f^V_N(J^\bul)-\E f^V_N(J^\bul)|>t\big)\le Ce^{-cN}.
\]
Explicit bounds are given in Proposition~\ref{prop:explicit-conc}.
\item \emph{Almost-sure limit.} Let $\bul\in\{G,R\}$ and $K(N)/N\to\alpha\in(0,\infty)$, and let
the matrices $J^\bul_{N,K(N)}$, $N\ge1$, be realized on a common
probability space with arbitrary dependence between volumes, and let
$\alpha_+:=\sup_NK(N)/N$. Then, for every deterministic $V:\Sigma_N\to\R$,
$\Var f^V_N(J^\bul_{N,K(N)})\le2\beta^2N^{-2}\E\norm{J^\bul_{N,K(N)}}_\op^2
\le C_{\alpha_+,\beta}/N$. There is an event of probability one on which
$\widehat A^\bul_{N,K(N)}(\beta,h)\to A^\bul(\alpha,\beta,h)$ for every
$(\beta,h)\in[0,\infty)\times\R$, locally uniformly. For fixed
$(\beta,h)$ the convergence also holds in $L^2$, and for every compact
$\mathcal C\subset[0,\infty)\times\R$, $\sup_{\mathcal C}|\widehat
A^\bul_{N,K(N)}-A^\bul|\to0$ in $L^1$.
\end{enumerate}
\end{theorem}

Let $\omega_\beta(\sigma):=e^{-\beta H_N(\sigma;J^\bul)}/Z^\bul_{N,K}
(\beta)$. The internal energy and entropy per spin and their quenched
averages are
\begin{equation}\label{eq:us}
u^\bul_N:=\frac1N\omega_\beta(H_N),\quad
s^\bul_N:=-\frac1N\sum_\sigma\omega_\beta(\sigma)\log\omega_\beta(\sigma),
\quad E^\bul_N:=\E u^\bul_N,\quad S^\bul_N:=\E s^\bul_N,
\end{equation}
with $K=K(N)$. The same quantities for~\eqref{eq:ZM} are denoted
$u^{\bul,M}_N$, $s^{\bul,M}_N$, $E^{\bul,M}_N$, $S^{\bul,M}_N$.
With the sign convention of~\eqref{eq:Ham}, $u^\bul_N\le0$.

\begin{theorem}[Internal energy and entropy]\label{thm:energy}
Let $\bul\in\{G,R\}$, $h=0$, $\alpha>0$ and $K(N)/N\to\alpha$.
\begin{enumerate}[label=\textup{(\roman*)}]
\item $A^\bul_{N,K(N)}(\cdot)$ is convex and $C^1$ on $[0,\infty)$,
with $\partial_\beta A^\bul_{N,K(N)}=-E^\bul_N$ and $A^\bul_{N,K(N)}=
-\beta E^\bul_N+S^\bul_N$.
\item $A^\bul(\alpha,\cdot)$ is convex, and the set $\mathcal N^\bul_\alpha$
of $\beta>0$ where it is not differentiable is at most countable. For
every $\beta\in(0,\infty)\setminus\mathcal N^\bul_\alpha$ the limits
$E^\bul(\alpha,\beta):=\lim_NE^\bul_N(\beta)$ and
$S^\bul(\alpha,\beta):=\lim_NS^\bul_N(\beta)$ exist and
\begin{equation}\label{eq:thermo}
E^\bul=-\partial_\beta A^\bul,\qquad A^\bul=-\beta E^\bul+S^\bul .
\end{equation}
On an event of probability one, $u^\bul_N(\beta)\to E^\bul(\alpha,
\beta)$ and $s^\bul_N(\beta)\to S^\bul(\alpha,\beta)$ simultaneously for
all $\beta\in(0,\infty)\setminus\mathcal N^\bul_\alpha$; for each such
$\beta$ the convergence also holds in $L^p$ for every $1\le p<\infty$.
\item For every $\beta>0$, the limit points of $E^\bul_N(\beta)$, and
almost surely those of $u^\bul_N(\beta)$, lie in
$[-\partial^+_\beta A^\bul(\alpha,\beta),-\partial^-_\beta
A^\bul(\alpha,\beta)]$; correspondingly those of $S^\bul_N(\beta)$ and
$s^\bul_N(\beta)$ lie in $[A^\bul-\beta\partial^+_\beta A^\bul,
A^\bul-\beta\partial^-_\beta A^\bul]$.
\item On $(0,\infty)\setminus\mathcal N^\bul_\alpha$, $E^\bul(\alpha,
\cdot)$ is nonincreasing, $E^\bul\le-\alpha/2$ and $0\le S^\bul\le
\log2$. Moreover $E^G\ge-\frac12(1+\sqrt\alpha)^2$.
\end{enumerate}
The same holds for the additional-pattern models, with $A^{\bul,M}$,
$E^{\bul,M}$, $S^{\bul,M}$ and the corresponding exceptional set
$\mathcal N^{\bul,M}_\alpha$, except for the bounds on the energy in
\textup{(iv)}.
\end{theorem}

\begin{theorem}[Additional binary pattern]\label{thm:Mattis}
Let $\bul\in\{G,R\}$, $\alpha>0$, $K(N)/N\to\alpha$ and $\beta\ge0$.
The limit $A^{\bul,M}(\alpha,\beta):=\lim_NA^{\bul,M}_{N,K(N)}(\beta)$
exists, does not depend on the law of $\eta$, and
\begin{equation}\label{eq:Mattislimit}
A^{\bul,M}(\alpha,\beta)=\max_{|v|\le1}\Big\{A^\bul(\alpha,\beta,\beta
v)-\frac{\beta v^2}2\Big\}=\sup_{v\in\R}\Big\{A^\bul(\alpha,\beta,\beta
v)-\frac{\beta v^2}2\Big\}.
\end{equation}
The random pressures $\widehat A^{\bul,M}_{N,K(N)}$ converge to
$A^{\bul,M}$ in $L^2$ and almost surely, locally uniformly in $\beta$,
for any joint realization of $(J^\bul_{N,K(N)},\eta_N)_{N\ge1}$. For
binary patterns, moreover,
\begin{equation}\label{eq:R-fixed-point}
A^{R,M}(\alpha,\beta)=A^R(\alpha,\beta);\quad\text{equivalently}\quad
A^R(\alpha,\beta,\beta v)\le A^R(\alpha,\beta)+\frac{\beta v^2}2
\quad(|v|\le1).
\end{equation}
\end{theorem}

The parameter $v$ in~\eqref{eq:Mattislimit} is a variational parameter
for the Mattis magnetization, i.e.\ the overlap with $\eta$; the formula
does not by itself establish a nonzero equilibrium overlap or
retrieval. Identity~\eqref{eq:R-fixed-point} reflects that an
additional binary pattern is one more binary pattern; it has no analogue for the Gaussian
bulk, where it would compare a binary pattern with Gaussian ones. This
asymmetry is quantified by the following consequence of
Theorems~\ref{thm:Rrep} and~\ref{thm:Mattis}.

\begin{corollary}[Comparison of the two pattern laws]\label{cor:compare}
For all $\alpha>0$, $\beta\ge0$ and $h\in\R$,
\[
A^G(\alpha,\beta,h)\le p_1(\alpha,\beta,h)\le A^R(\alpha,\beta,h),
\qquad p_1(\alpha,\beta,0)=A^{G,M}(\alpha,\beta).
\]
In particular $A^G(\alpha,\beta)\le A^{G,M}(\alpha,\beta)\le A^R(\alpha,
\beta)=A^{R,M}(\alpha,\beta)$, and $\sup_rp_r(\alpha,\beta,0)=p_1(\alpha,
\beta,0)$ if and only if $A^R(\alpha,\beta)=A^{G,M}(\alpha,\beta)$.
\end{corollary}

Thus ``one retained pattern suffices'' at zero field is equivalent to
the statement that the binary Hopfield pressure equals that of a single
binary pattern added to a Gaussian bulk, in line with the standard
heuristic in which non-condensed patterns act as Gaussian
noise~\cite{AGS87,ABLT}. The identification $p_1=A^{G,M}$ is restricted to
$h=0$: Theorem~\ref{thm:Mattis} is formulated without a uniform field,
and for $h\ne0$ the gauge transformation of Lemma~\ref{lem:gauge}, which
turns the additional pattern into a Curie--Weiss term, turns the uniform
field into the pattern-dependent field $h\sum_i\eta_i\sigma_i$.

\begin{proposition}[Properties of the limits]\label{prop:properties}
Let $\bul\in\{G,R\}$.
\begin{enumerate}[label=\textup{(\alph*)}]
\item $h\mapsto A^\bul(\alpha,\beta,h)$ is convex, even and
$1$-Lipschitz; $\beta\mapsto A^\bul(\alpha,\beta,h)$ is convex and
nondecreasing; $\alpha\mapsto A^\bul(\alpha,\beta,h)-\beta\alpha/2$ is
nondecreasing, and so is $\alpha\mapsto A^\bul(\alpha,\beta,h)$.
\item $A^\bul(\alpha,\beta,h)\ge\log(2\cosh h)+\frac{\beta\alpha}2$, and,
if $\beta<1$, $A^\bul(\alpha,\beta,h)\le\log(2\cosh h)-\frac\alpha2
\log(1-\beta)$.
\item $A^G(\alpha,\beta,h)\le\log(2\cosh h)+\frac\beta2(1+\sqrt\alpha)^2$,
and for $0<\alpha_1\le\alpha_2$,
$A^G(\alpha_2,\beta,h)-A^G(\alpha_1,\beta,h)\le\frac\beta2
(1-\frac{\alpha_1}{\alpha_2})(1+\sqrt{\alpha_2})^2$.
\end{enumerate}
\end{proposition}

\begin{remark}[Diagonal terms]\label{rem:diag}
Removing the terms $i=j$ subtracts $T_N:=\beta\norm{J^\bul}_\HS^2/(2N^2)$
from $\widehat A^\bul_{N,K}$, $T_N+\beta/(2N)$ from $\widehat
A^{\bul,M}_{N,K}$, and $T^{(r)}_N:=\frac\beta{2N^2}(Nr+\norm{G'}_\HS^2)$
from the mixed pressure $N^{-1}\log Z^{(r)}_{N,K}$, where $\E
T^{(r)}_N=\beta K/(2N)$ and $T^{(r)}_N\to\beta\alpha/2$ almost surely, as
for $T_N$ below; equivalently, it adds
$\norm{J^\bul}_\HS^2/(2N^2)$, respectively
$\norm{J^\bul}_\HS^2/(2N^2)+1/(2N)$, to the energy per spin. For binary
patterns $T_N=\beta K/(2N)$ is deterministic. For Gaussian patterns
$\E T_N=\beta K/(2N)$ and 
$$
\E(T_N-\E T_N)^4=(\beta/2N^2)^4[60NK+
12NK(NK-1)]\le C(\beta,\alpha_+)N^{-4},
$$
with $\alpha_+:=\sup_NK(N)/N$,
since $T_N$ is a sum of $NK$ i.i.d.\ variables with $\E(Y^2-1)^4=60$;
hence $T_N\to\beta\alpha/2$ almost surely and in every $L^p$, $1\le p<\infty$. Write
${\rm off}$ for quantities without diagonal terms.
\begin{enumerate}[label=\textup{(\roman*)},leftmargin=*]
\item \emph{Transformation formulas.}
\[
A^{\bul,{\rm off}}=A^\bul-\frac{\beta\alpha}2,\quad
p_r^{\rm off}=p_r-\frac{\beta\alpha}2,\quad
A^{\bul,M,{\rm off}}=A^{\bul,M}-\frac{\beta\alpha}2,\quad
g^{\bul,{\rm off}}=g^\bul-\frac\alpha2,
\]
\[
E^{\bul,{\rm off}}=E^\bul+\frac\alpha2,\quad
E^{\bul,M,{\rm off}}=E^{\bul,M}+\frac\alpha2,\quad
S^{\bul,{\rm off}}=S^\bul,\quad S^{\bul,M,{\rm off}}=S^{\bul,M},
\]
where $-g^\bul$ is the ground-state energy of
Corollary~\ref{cor:ground}.
\item \emph{Transfer of conclusions.} Existence and sequence
independence of the limits, local uniformity, almost-sure, $L^p$
($1\le p<\infty$) and
self-averaging statements in Theorems~\ref{thm:main}--\ref{thm:Mattis}
and Corollary~\ref{cor:ground} hold for the off-diagonal quantities.
The shifts are linear in $\beta$, so differentiability points, the
exceptional sets, the identities~\eqref{eq:thermo}, the one-sided bounds
of Theorem~\ref{thm:energy}(iii), the formulas~\eqref{eq:main},
\eqref{eq:Mattislimit} and~\eqref{eq:R-fixed-point} (with off-diagonal
quantities on both sides) and the inequalities of
Corollary~\ref{cor:compare} are unchanged. Exponential concentration
(Theorem~\ref{thm:concentration}(a)) also persists: for binary patterns
the shift is deterministic, and for Gaussian patterns $T_N-\E T_N$
satisfies Bernstein's inequality for sums of $NK$ independent
sub-exponential variables~\cite{BLM13}. The explicit constants of
Proposition~\ref{prop:explicit-conc}(a) are not claimed for the
off-diagonal Gaussian model.
\item \emph{Transformed inequalities.} The bounds of
Proposition~\ref{prop:properties} and Theorem~\ref{thm:energy}(iv)
become
\[
\log(2\cosh h)\le A^{\bul,{\rm off}}(\alpha,\beta,h)\le\log(2\cosh h)
-\frac\alpha2\big[\log(1-\beta)+\beta\big]\quad(\beta<1),
\]
\[
A^{G,{\rm off}}(\alpha,\beta,h)\le\log(2\cosh h)+\frac\beta2(1+2\sqrt
\alpha),\qquad
-\frac12-\sqrt\alpha\le E^{G,{\rm off}}\le0,\qquad E^{R,{\rm off}}\le0,
\]
and, for $0<\alpha_1\le\alpha_2$,
\[
0\le A^{G,{\rm off}}(\alpha_2,\beta,h)
-A^{G,{\rm off}}(\alpha_1,\beta,h)\le\frac\beta2\Big[\Big(1-\frac{\alpha_1}
{\alpha_2}\Big)(1+\sqrt{\alpha_2})^2-(\alpha_2-\alpha_1)\Big].
\]
The properties in $h$, convexity and monotonicity in $\beta$ (since
$E^{\bul,{\rm off}}\le0$), monotonicity of $E^{\bul,{\rm off}}$ and the
bounds $0\le S\le\log2$ are unchanged, and
$\alpha\mapsto A^{\bul,{\rm off}}(\alpha,\beta,h)$ is nondecreasing by
Proposition~\ref{prop:properties}(a).
\end{enumerate}
\end{remark}

\begin{remark}[Relation with~\cite{CCM}]\label{rem:CCM}
In the notation of~\cite{CCM}, $p_N(\beta)=\widehat A^G_{N,K}(\beta)-
\log2$, with diagonal terms retained, so their limit $P(\alpha,\beta)$
is $A^G(\alpha,\beta)-\log2$, and Theorem~\ref{thm:main} gives an
alternative proof of the existence part of their result.
Proposition~\ref{prop:properties}(b) is consistent with their
statements: $P\le P_{\rm ann}=-\frac\alpha2\log(1-\beta)$ for
$\beta<1$, with equality for $\beta\le(1+\sqrt\alpha)^{-1}$
\cite{BGGT12,CCM}, and $P$ is finite for every $\beta$. The rotational
invariance behind Lemma~\ref{lem:rotation} is also the mechanism for
the absence of retrieval of Gaussian patterns noted in~\cite{CCM}.
\end{remark}

\section{Preliminaries}\label{sec:prelim}

We collect the concentration, convexity and matrix estimates used below.
In \ref{T:Poincare}--\ref{T:McShane}, $Y$ denotes a standard Gaussian
vector in $\R^n$.
\begin{enumerate}[label=\textup{(T\arabic*)},leftmargin=*]
\item\label{T:Poincare} \emph{Gaussian Poincar\'e inequality}
\cite[Thm.~3.20]{BLM13}: $\Var\varphi(Y)\le\E\norm{\nabla\varphi(Y)}^2$
for $\varphi\in C^1$ with $\varphi(Y),\nabla\varphi(Y)\in L^2$, and, by
approximation, for Lipschitz $\varphi$.
\item\label{T:Lipschitz} \emph{Gaussian concentration}
(Tsirelson--Ibragimov--Sudakov) \cite[Thm.~5.6]{BLM13}: if $\varphi$ is
$L$-Lipschitz, then
$\Prob(\varphi(Y)-\E\varphi(Y)\ge t)\le e^{-t^2/(2L^2)}$ and
$\Prob(|\varphi(Y)-\E\varphi(Y)|\ge t)\le2e^{-t^2/(2L^2)}$.
\item\label{T:McShane} \emph{McShane extension} \cite{McS34}: if
$\varphi:B\to\R$ is $L$-Lipschitz on $B\subset\R^n$, then
$x\mapsto\inf_{y\in B}\{\varphi(y)+L\norm{x-y}\}$ is $L$-Lipschitz on
$\R^n$ and equals $\varphi$ on $B$.
\item\label{T:BC} \emph{Borel--Cantelli}: if $\sum_N\Prob(E_N)<\infty$,
then almost surely only finitely many $E_N$ occur (no independence is
needed).
\item\label{T:convex} \emph{Convex functions} \cite[Thms.~10.1, 10.8,
24.1]{Roc70}: a finite convex function on an open interval has
nondecreasing one-sided derivatives $f'_-\le f'_+$, with
$f'_+(x)\le f'_-(y)$ for $x<y$, hence is differentiable outside a
countable set; and convex functions converging pointwise to a finite
limit on a dense subset converge everywhere, locally uniformly, to a
finite convex function.
\end{enumerate}

The next lemma transfers pressure convergence to convergence of the
derivative in Section~\ref{sec:energy}. Applied to differentiation in
$\beta$, it gives the internal-energy limit and one-sided bounds at
nondifferentiability points.

\begin{lemma}\label{lem:convexderiv}
Let $f_N,f$ be convex on an open interval $I$, with $f_N$ differentiable
and $f_N\to f$ pointwise. Then
$f'_-(x)\le\liminf_Nf_N'(x)\le\limsup_Nf_N'(x)\le f'_+(x)$ for every
$x\in I$.
\end{lemma}
\begin{proof}
For $y<x<z$ in $I$, convexity gives $\frac{f_N(x)-f_N(y)}{x-y}\le
f_N'(x)\le\frac{f_N(z)-f_N(x)}{z-x}$. Let $N\to\infty$, then $y\uparrow
x$ and $z\downarrow x$.
\end{proof}

\subsection{Random matrices}

The statistical pressure is controlled by the operator norm of the
pattern matrix, since $0\le-H_N(\sigma;J)\le\frac12\norm J_\op^2$. We
first collect the moment and tail bounds for Gaussian matrices, which
are used in the uniform truncation estimate (Section~\ref{sec:cut}) and
in the concentration bounds for Gaussian patterns.

\begin{lemma}[Gaussian matrices]\label{lem:norm}
For $N\ge1$, $K\ge0$: $\E\norm G_\op\le\sqrt N+\sqrt K$,
$\E\norm G_\op^2\le(\sqrt N+\sqrt K)^2+1\le3(N+K)$ and
$\E\norm G_\op^4\le8(\sqrt N+\sqrt K)^4+128\le160(N+K)^2$. Moreover
$\Prob(\norm G_\op>2(\sqrt N+\sqrt K))\le e^{-(\sqrt N+\sqrt K)^2/2}$.
\end{lemma}
\begin{proof}
The first bound is the Slepian--Gordon inequality
\cite[Thm.~II.13]{DS01}, \cite[Thm.~7.3.1]{Ver18}. Since
$G\mapsto\norm G_\op$ is $1$-Lipschitz for the Hilbert--Schmidt norm,
$Z:=\norm G_\op-\E\norm G_\op$ satisfies $\Var Z\le1$ by
\ref{T:Poincare} and $\Prob(|Z|>t)\le2e^{-t^2/2}$ by
\ref{T:Lipschitz}, whence $\E Z^4\le16$ and the tail bound. Use
$(x+y)^4\le8x^4+8y^4$, $(\sqrt N+\sqrt K)^2\le2(N+K)$ and $N+K\ge1$.
\end{proof}

For binary and mixed pattern matrices, a net argument gives the
following bound without rotational invariance. It applies to any matrix with independent Gaussian or
Rademacher entries; the constant $c_\star$ below fixes the radius of the
high-probability event used in the radial cutoff
(Lemma~\ref{lem:radial}) and in the concentration bounds for binary
patterns.

\begin{lemma}[Sub-Gaussian matrices]\label{lem:subG}
Let $J$ be an $N\times K$ matrix with independent entries, each
standard Gaussian or symmetric Rademacher. Then
\begin{equation}\label{eq:net}
\Prob(\norm J_\op>s)\le\min\{1,2\cdot9^{N+K}e^{-s^2/8}\},\qquad
\E\norm J_\op^q\le C_q(N+K)^{q/2},\quad1\le q<\infty .
\end{equation}
In particular, if $K\le\alpha_+N$ and $c_\star:=\big(8\{1+(1+\alpha_+)
\log9\}\big)^{1/2}$ \textup{(}as in
Proposition~\textup{\ref{prop:explicit-conc}(b))}, then
$\Prob(\norm J_\op>c_\star\sqrt N)\le2e^{-N}$ for every $N$.
\end{lemma}

\begin{proof}
For unit vectors $u,v$, the bilinear form $u^TJv=\sum_{i,\mu}u_iv_\mu
J_i^\mu$ is centered sub-Gaussian with $\Prob(|u^TJv|>t)\le2e^{-t^2/2}$.
For $1/4$-nets $\mathcal N_1,\mathcal N_2$ of the unit spheres of $\R^N$
and $\R^K$, of cardinalities at most $9^N$ and $9^K$, one has
$\norm J_\op\le2\max_{u\in\mathcal N_1,v\in\mathcal N_2}u^TJv$
\cite[Ex.~4.4.3]{Ver18}. A union bound gives the tail estimate, and
integrating it gives the moment bound. For the last claim,
$2\cdot9^{N+K}e^{-c_\star^2N/8}\le2\exp\{(N+K)\log9-N-(1+\alpha_+)
N\log9\}\le2e^{-N}$.
\end{proof}

\subsection{Elementary properties of the statistical pressure}

We next record the analytic properties of $f^V_N$ in the disorder and
the parameters. Lemma~\ref{lem:gradient} shows that $f^V_N$ is convex in
the disorder, with its gradient controlled by the operator norm. These
properties underlie the concentration estimates for both pattern laws. Lemma~\ref{lem:elementary} collects the
monotonicity, convexity and Lipschitz properties in $\beta$ and $h$, the
a priori bounds on the quenched pressure, and the annealed bound, which
is used for Proposition~\ref{prop:properties} and for the equicontinuity
arguments that upgrade pointwise to locally uniform convergence.

\begin{lemma}\label{lem:gradient}
Let $V:\Sigma_N\to\R$. The function $J\mapsto f^V_N(J)$ on
$\R^{N\times K}$ is smooth and convex, $\nabla f^V_N(0)=0$, and
\begin{equation}\label{eq:gradient}
\norm{\nabla f^V_N(J)}_\HS\le\frac\beta N\norm J_\op,\qquad
|f^V_N(J)|\le c_V+\frac\beta{2N}\norm J_\op^2 .
\end{equation}
\end{lemma}

\begin{proof}
Each exponent $\frac\beta{2N}\norm{J^T\sigma}^2+V(\sigma)$ is convex in
$J$, and log-sum-exp preserves convexity. With $\omega$ the Gibbs
measure, $\nabla f^V_N=\frac\beta{N^2}\omega(\sigma(J^T\sigma)^T)$, which
vanishes at $J=0$, and Jensen's inequality for $\norm\cdot_\HS^2$ gives
$\norm{\nabla f^V_N}_\HS^2\le\frac{\beta^2}{N^4}\omega(N\norm{J^T\sigma}^2)
\le\frac{\beta^2}{N^2}\norm J_\op^2$. The second bound follows from
$0\le-H_N(\sigma;J)\le\frac12\norm J_\op^2$.
\end{proof}

\begin{lemma}\label{lem:elementary}
Let $J$ be a real $N\times K$ matrix and $\widehat A(\beta,h):=
f^{hNm}_N(J)$. Pointwise in $J$, $\log(2\cosh h)\le\widehat A\le
\log(2\cosh h)+\frac\beta{2N}\norm J_\op^2$ and $|\widehat A(\beta,h)-
\widehat A(\beta,h')|\le|h-h'|$; $\widehat A$ is convex and even in $h$,
and convex and nondecreasing in $\beta$, with $0\le\partial_\beta
\widehat A\le\norm J_\op^2/(2N)$. For $\bul\in\{G,R\}$,
\begin{equation}\label{eq:meanbounds}
\log(2\cosh h)+\frac{\beta K}{2N}\le A^\bul_{N,K}(h)\le\log(2\cosh h)
+\frac\beta{2N}\E\norm{J^\bul}_\op^2,
\end{equation}
with $\frac1N\E\norm G_\op^2\le(1+\sqrt{K/N})^2+\frac1N$, and
$A^\bul_{N,K}(h)\le\log(2\cosh h)-\frac K{2N}\log(1-\beta)$ if $\beta<1$.
\end{lemma}

\begin{proof}
The pointwise bounds follow from $0\le-H_N\le\frac12\norm J_\op^2$.
With $\omega$ the Gibbs measure, $\partial_h\widehat A=\omega(m)\in
[-1,1]$, $\partial_h^2\widehat A=N\Var_\omega(m)\ge0$, $\partial_\beta
\widehat A=-N^{-1}\omega(H_N)$ and $\partial_\beta^2\widehat A=N^{-1}
\Var_\omega(H_N)\ge0$; evenness follows from $\sigma\mapsto-\sigma$. The
upper bound in~\eqref{eq:meanbounds} is the pointwise one, and the bound
on $\E\norm G_\op^2$ is Lemma~\ref{lem:norm}. For the lower bound, with
$\nu_h(\pm1):=e^{\pm h}/(2\cosh h)$, Jensen's inequality gives
$\widehat A(\beta,h)\ge\log(2\cosh h)+\frac\beta{2N^2}\nu_h^{\otimes N}
(\norm{J^T\sigma}^2)$, and $\E\norm{(J^\bul)^T\sigma}^2=KN$ for every
$\sigma$, since the entries are independent, centered and of unit
variance. Finally, for fixed $\sigma$ the variables $X_\mu:=N^{-1/2}
(J^\bul)^\mu\cdot\sigma$ are independent and satisfy $\E e^{\lambda
X_\mu}\le e^{\lambda^2/2}$ (with equality in the Gaussian case, and
since $\cosh(\lambda/\sqrt N)^N\le e^{\lambda^2/2}$ in the binary case).
For $\beta<1$, with $Y$ an independent standard Gaussian, $\E
e^{\beta X_\mu^2/2}=\E e^{\sqrt\beta YX_\mu}\le\E e^{\beta Y^2/2}=
(1-\beta)^{-1/2}$. Hence $\E Z^\bul_{N,K}(h)\le(2\cosh h)^N
(1-\beta)^{-K/2}$, and $\E\log Z\le\log\E Z$.
\end{proof}

\subsection{Vector spin glasses with compact spins}

For Gaussian patterns, the required external convergence theorem concerns
vector spin glasses with compactly supported spins and possibly
non-convex covariance. After blocking, the Hubbard--Stratonovich representation of
the Hopfield model takes this form, except that the hidden spins are
unbounded; the truncation estimate of Section~\ref{sec:cut} removes this
obstruction. We state the result in the form used
in Section~\ref{sec:rational}.

For $S,T\in\R^{D\times n}$ with columns $S_j,T_j$, let
$R_n(S,T):=n^{-1}\sum_{j=1}^nS_jT_j^T$, and let
$\mathsf B_D:=\{x\in\R^D:\norm x\le1\}$. We use~\cite[Theorem~1.1]{CHI}
with the auxiliary external parameter set to zero. Existence of the
limit was first proved in~\cite[Theorem~1.1]{Ho26} for any compactly
supported $\mu$, without the restriction to the unit ball used below;
the limit is identified in~\cite{CHI}.

\begin{theorem}[Ho~\cite{Ho26}; Chen--Ho--Issa~\cite{CHI}]\label{thm:CHI}
Let $\mu$ be a probability measure on $\R^D$ with $\mu(\mathsf B_D)=1$,
and let $\psi:\R^{D\times D}\to\R$ be given by an everywhere absolutely
convergent power series with $\psi(0)=0$. Assume that for each $n$
there is a centered Gaussian process $(X_n(S))_{S\in\R^{D\times n}}$
with $\E X_n(S)X_n(T)=n\psi(R_n(S,T))$. Then for every $t\ge0$ the
limit of
\begin{equation}\label{eq:CHI}
\Phi_n(t):=\frac1n\E\log\int\exp\Big\{\sqrt{2t}X_n(S)-tn\psi(R_n(S,S))
\Big\}\mu^{\otimes n}(dS)
\end{equation}
exists and is finite. Convexity of $\psi$ is not assumed.
\end{theorem}

In~\cite{CHI} the theorem concerns $-\Phi_n(t)$, which the authors call
the free energy; they also give a variational formula for the limit.
Here we use only its existence, at $t=\frac12$.

\subsection{Concentration on the discrete cube}
For binary patterns, concentration on the discrete cube replaces the
Gaussian inequalities \ref{T:Poincare}--\ref{T:Lipschitz}. The relevant
functions are convex in the disorder (Lemma~\ref{lem:gradient}), so
three dimension-free estimates suffice.
Lemma~\ref{lem:cubevar} bounds the variance and gives $L^2$
self-averaging for binary patterns. Lemma~\ref{lem:convex} bounds the
upper tail and the expected maximum of finitely many, possibly dependent,
convex functions; it controls the disorder-dependent choice of the
retained patterns in the sector selection (Section~\ref{sec:selection}).
Lemma~\ref{lem:two-sided}, based on Talagrand's convex-distance
inequality, gives two-sided exponential concentration, from which
almost-sure convergence follows.

Let $X$ be uniform on $\{-1,1\}^d$, and let $x^{(i)}$ denote $x$ with
its $i$-th coordinate flipped.

\begin{lemma}[Variance of convex functions]\label{lem:cubevar}
If $\varphi:\R^d\to\R$ is convex and differentiable, then
$\Var\varphi(X)\le2\E\norm{\nabla\varphi(X)}^2$.
\end{lemma}

\begin{proof}
By the Efron--Stein inequality and invariance of the law of $X$ under
coordinate flips,
$\Var\varphi\le\frac14\E\sum_i(\varphi-\varphi^{(i)})^2=\frac12\E\sum_i
(\varphi-\varphi^{(i)})_+^2$, where $\varphi^{(i)}:=\varphi(X^{(i)})$.
Convexity gives $\varphi(x^{(i)})\ge\varphi(x)-2x_i\partial_i\varphi(x)$,
so $(\varphi-\varphi^{(i)})_+\le2|\partial_i\varphi|$.
\end{proof}

\begin{lemma}[Convex upper concentration]\label{lem:convex}
If $\varphi:\R^d\to\R$ is differentiable and convex with
$\norm{\nabla\varphi}\le L$ everywhere, then for $\lambda\ge0$
\[
\log\E e^{\lambda(\varphi(X)-\E\varphi(X))}\le2L^2\lambda^2 .
\]
Consequently, for $m$ such functions $\varphi_1,\dots,\varphi_m$, not
necessarily independent,
\[
\E\max_{j\le m}\varphi_j(X)\le\max_j\E\varphi_j(X)+2\sqrt2\,L
\sqrt{\log(m+1)} .
\]
\end{lemma}

\begin{proof}
For a two-point uniform law, Jensen's inequality gives
$\Ent(e^g)\le\operatorname{Cov}(e^g,g)$. For $g=\lambda\varphi$ with
values $\lambda a\ge\lambda b$ this covariance is
$\frac14(e^{\lambda a}-e^{\lambda b})\lambda(a-b)\le\frac14\lambda^2
e^{\lambda a}(a-b)^2$, which equals $\frac{\lambda^2}2$ times the
conditional expectation of $e^{\lambda\varphi}(\varphi-\varphi^{(i)})_+^2$.
Tensorization of entropy~\cite{BLM13} gives
$\Ent(e^{\lambda\varphi})\le\frac{\lambda^2}2\E e^{\lambda\varphi}
\sum_i(\varphi-\varphi^{(i)})_+^2\le2\lambda^2L^2\E e^{\lambda\varphi}$,
by the convexity bound of Lemma~\ref{lem:cubevar}. With
$\Lambda(\lambda):=\log\E e^{\lambda(\varphi-\E\varphi)}$ this reads
$(\lambda\Lambda'-\Lambda)/\lambda^2\le2L^2$; integrating
$(\Lambda/\lambda)'$ from $0$ proves the first claim. For the second,
$\E\max_j(\varphi_j-\E\varphi_j)\le\lambda^{-1}\log\sum_j\E
e^{\lambda(\varphi_j-\E\varphi_j)}\le\lambda^{-1}\log m+2L^2\lambda$;
optimize in $\lambda$.
\end{proof}

\begin{lemma}[Two-sided convex concentration]\label{lem:two-sided}
If $\varphi:\R^d\to\R$ is convex and $L$-Lipschitz, $L>0$, and $m_\varphi$
is a median of $\varphi(X)$, then for $t>0$
\begin{equation}\label{eq:two-sided}
\Prob(|\varphi(X)-m_\varphi|>t)\le4\exp\{-t^2/(16L^2)\},\qquad
|\E\varphi(X)-m_\varphi|\le8\sqrt\pi L .
\end{equation}
\end{lemma}

\begin{proof}
For $B\subseteq\{-1,1\}^d$ let $d_T(x,B):=\sup\{\inf_{y\in B}\sum_i
a_i\ind_{\{x_i\ne y_i\}}:a_i\ge0,\ \sum_ia_i^2\le1\}$. Talagrand's
convex-distance inequality~\cite{Tal95} states
$\Prob(B)\Prob(d_T(X,B)\ge u)\le e^{-u^2/4}$. Every subgradient $v_x$ of
$\varphi$ at $x$ has norm at most $L$, and $\varphi(x)-\varphi(y)\le
v_x\cdot(x-y)\le2\sum_i|v_{x,i}|\ind_{\{x_i\ne y_i\}}$; hence, if
$\varphi(x)-\varphi(y)\ge t$ for every $y\in B$, then $d_T(x,B)\ge
t/(2L)$. Apply the inequality with $B=\{\varphi\le m_\varphi\}$ for the
upper tail, and with $B=\{\varphi\le m_\varphi-t\}$, using
$\Prob(\varphi\ge m_\varphi)\ge\frac12$, for the lower tail. Each tail
is at most $2e^{-t^2/(16L^2)}$, and integrating gives $\E|\varphi-
m_\varphi|\le8\sqrt\pi L$.
\end{proof}

\section{Gaussian bulks with finitely many site types}\label{sec:gauss}

We first control the unbounded Gaussian hidden spins, then apply the
compact-spin theorem to a block representation. Field, volume and load
comparisons extend the resulting limit beyond the block sequences.

\subsection{Uniform truncation of the hidden spins}\label{sec:cut}
Let $V:\Sigma_N\to\R$ be arbitrary. Write $\omega^V_G$ for the Gibbs
measure on $\Sigma_N$ with weights $e^{-\beta H_N(\sigma;G)+V(\sigma)}$,
let $P^V(dG,\sigma):=\ga^{\otimes NK}(dG)\,\omega^V_G(\sigma)$ be the
joint law of disorder and spins, $\E\omega^V_G$ the corresponding
expectation, and set $a:=\sqrt{\beta/N}\,G^T\sigma\in\R^K$. Since
$\int_{\R^K}e^{x\cdot z}\ga^{\otimes K}(dz)=e^{\norm x^2/2}$ and
$\norm a^2/2=-\beta H_N(\sigma;G)$ (Hubbard--Stratonovich),
\begin{equation}\label{eq:HS}
Z^V:=\sum_\sigma e^{-\beta H_N(\sigma;G)+V(\sigma)}=\sum_\sigma
e^{V(\sigma)}\int_{\R^K}e^{\sqrt{\beta/N}\,\sigma^TGz}\,
\ga^{\otimes K}(dz).
\end{equation}
For $L>0$ let $Z^{V,[L]}$ be the same expression with the integral over
$[-L,L]^K$, without renormalizing the Gaussian measure. If $V$ is random
and independent of $G$, all statements below hold conditionally on $V$,
and for the unconditional expectations whenever
$\E\max_\sigma|V(\sigma)|<\infty$.

\begin{lemma}[Rotational invariance]\label{lem:rotation}
Let $N,K\ge1$. Under $P^V$ the law of $a$ is $O(K)$-invariant, hence
exchangeable, and
$\E\omega^V_G(a_1^4)=\frac3{K(K+2)}\E\omega^V_G(\norm a^4)\le
\frac{3\beta^2}{K(K+2)}\E\norm G_\op^4$. If $K\ge\alpha_-N$, then
$\E\omega^V_G(a_1^4)\le M_4:=480\beta^2(1+\alpha_-^{-1})^2$, uniformly
in $N,K,V$.
\end{lemma}

\begin{proof}
For $O\in O(K)$, the map $G\mapsto GO$ preserves $\ga^{\otimes NK}$ (the
rows of $G$ are i.i.d.\ standard Gaussian in $\R^K$) and $H_N$; since
$V$ depends only on $\sigma$, it preserves $\omega^V_G$ and hence
$P^V$, while $a\mapsto O^Ta$. If $\mathcal O$ is Haar-distributed and
independent, $a_1$ has the law of $a\cdot U$ with $U=\mathcal Oe_1$
uniform on the sphere and independent of $a$, and $\E U_1^4=3/(K(K+2))$.
Finally $\norm a^2\le\beta\norm G_\op^2$, and Lemma~\ref{lem:norm} gives
$\frac{3\beta^2}{K^2}160(N+K)^2\le M_4$.
\end{proof}

\begin{proposition}[Uniform truncation]\label{prop:cut}
Let $0<\alpha_-\le\alpha_+$ and $C_{\rm cut}:=\alpha_+(32/e+6M_4)$. For
all $L\ge4$, $N,K\ge1$ with $K/N\in[\alpha_-,\alpha_+]$, and all $V$,
\begin{equation}\label{eq:cut}
0\le\frac1N\E\log Z^V-\frac1N\E\log Z^{V,[L]}\le\frac{C_{\rm cut}}{L^2}.
\end{equation}
\end{proposition}

\begin{proof}
Let $\pi_L(v):=\Prob(|Y+v|\le L)$ and $f_L:=-\log\pi_L\ge0$. Since
$\int_{-L}^Le^{xz}\ga(dz)=e^{x^2/2}\pi_L(x)$,
\[
Z^{V,[L]}/Z^V=\omega^V_G\Big(\exp\Big\{-\sum_{\mu\le K}f_L(a_\mu)\Big\}
\Big)\in(0,1].
\]
The lower bound follows. For the upper bound, Jensen's inequality and
exchangeability give
$\frac1N\E\log(Z^V/Z^{V,[L]})\le\frac KN\E\omega^V_G(f_L(a_1))$. If
$|v|\le L/2$, then $1-\pi_L(v)\le\Prob(|Y|>L/2)\le2e^{-L^2/8}\le
2e^{-2}$, and $-\log(1-x)\le2x$ on $[0,\frac12]$ gives $f_L(v)\le
4e^{-L^2/8}$. For all $v$, $\pi_L(v)\ge\int_{-1}^1\frac{e^{-(z-v)^2/2}}
{\sqrt{2\pi}}dz\ge\frac2{\sqrt{2\pi}}e^{-(|v|+1)^2/2}$, so $f_L(v)\le
v^2+2$. By Markov's inequality and Lemma~\ref{lem:rotation},
\[
\E\omega^V_G(f_L(a_1))\le4e^{-L^2/8}+\E\omega^V_G\big[(a_1^2+2)
\ind_{\{|a_1|>L/2\}}\big]\le4e^{-L^2/8}+\Big(\frac4{L^2}+\frac{32}{L^4}
\Big)M_4\le\frac{32/e+6M_4}{L^2},
\]
using $e^{-x}\le(ex)^{-1}$ and $L\ge4$.
\end{proof}

\subsection{Site fields and blocking}\label{sec:rational}
For a field vector $\vec h=(h_1,\dots,h_N)\in\R^N$ let
\begin{equation}\label{eq:Bdef}
B_{N,K}(\vec h):=\frac1N\E\log\sum_{\sigma\in\Sigma_N}\exp\Big\{
\frac\beta{2N}\norm{G^T\sigma}^2+\sum_{i=1}^Nh_i\sigma_i\Big\}
=\E f^{V_{\vec h}}_N(G),\qquad V_{\vec h}(\sigma):=\sum_ih_i\sigma_i .
\end{equation}
Fix $D\ge1$ and fields $H=(H_1,\dots,H_D)\in\R^D$. A \emph{type
assignment} is a map $\mathsf t:\{1,\dots,N\}\to\{1,\dots,D\}$, and
$H_{\mathsf t}:=(H_{\mathsf t_1},\dots,H_{\mathsf t_N})$. With $D=1$ and $H_1=h$,
$B_{N,K}(H_{\mathsf t})=A^G_{N,K}(\beta,h)$.

Fix $a,b\in\N$, $\mathsf D:=a+b$, a map $\tau:\{1,\dots,b\}\to
\{1,\dots,D\}$, and let $N=bn$, $K=an$. Relabel
$\sigma_i=\sigma_{j,u}$ ($i=(j-1)b+u$, $u\le b$), $z_\mu=z_{k,v}$
($\mu=(k-1)a+v$, $v\le a$) and $g_{(j,u),(k,v)}:=g_i^\mu$, and give site
$(j,u)$ the type $\tau(u)$; we call this the \emph{block assignment}.
For $S\in\R^{\mathsf D\times n}$ with columns $S_j=(x_{j,1},\dots,
x_{j,b},y_{j,1},\dots,y_{j,a})$ define
\begin{equation}\label{eq:Xn}
X_n(S):=\sqrt{\frac\beta{bn}}\sum_{j,k\le n}\sum_{u\le b,\,v\le a}
g_{(j,u),(k,v)}x_{j,u}y_{k,v},\qquad
\Psi(R):=\frac\beta b\Big(\sum_{u\le b}R_{uu}\Big)\Big(\sum_{v=b+1}^{\mathsf D}
R_{vv}\Big).
\end{equation}
Then $X_n$ is a centered Gaussian process on all of $\R^{\mathsf D\times
n}$ with $\E X_n(S)X_n(T)=n\Psi(R_n(S,T))$, since $\E X_n(S)X_n(T)=
\frac\beta{bn}(\sum_{j,u}x_{j,u}x'_{j,u})(\sum_{k,v}y_{k,v}y'_{k,v})$.
The polynomial $\Psi$ is not convex: on $\mathrm{diag}(sI_b,tI_a)$ it
equals $\beta a\,st$. Let $\nu_h$ be as in Lemma~\ref{lem:elementary} and
\[
\lambda_{L}(dz):=\frac{\ind_{[-L,L]}(z)e^{\beta z^2/2}}{d_L}\ga(dz),
\quad d_L:=\int_{-L}^Le^{\beta z^2/2}\ga(dz),\quad
\mu_{L}:=\Big(\bigotimes_{u\le b}\nu_{H_{\tau(u)}}\Big)\otimes
\lambda_L^{\otimes a},
\]
a probability measure on $\R^{\mathsf D}$ supported in the ball of
radius $c:=\sqrt{b+aL^2}$, for every $\beta\ge0$. (Without truncation,
$e^{\beta z^2/2}\ga(dz)$ has infinite mass for $\beta\ge1$.) Let
$Z^{[L]}$ denote the truncated partition function~\eqref{eq:HS} with
$V=V_{H_{\mathsf t}}$ for the block assignment.

\begin{lemma}\label{lem:block}
For the block assignment,
\begin{equation}\label{eq:cancel}
Z^{[L]}=\prod_{u\le b}(2\cosh H_{\tau(u)})^{n}\,d_L^{an}\int\exp\Big\{
X_n(S)-\frac n2\Psi(R_n(S,S))\Big\}\mu_{L}^{\otimes n}(dS).
\end{equation}
\end{lemma}

\begin{proof}
When $(\sigma,z)$ is arranged into $S$ as above, $X_n(S)=\sqrt{\beta/N}
\,\sigma^TGz$. If all $x_{j,u}\in\{\pm1\}$, as holds
$\mu_L^{\otimes n}$-a.s., then $\sum_uR_n(S,S)_{uu}=b$, so that
$\frac n2\Psi(R_n(S,S))=\frac\beta2\sum_{k,v}z_{k,v}^2$. The claim
follows from $\ind_{[-L,L]}(z)\ga(dz)=d_Le^{-\beta z^2/2}\lambda_L(dz)$
and $\sum_{s=\pm1}e^{hs}\varphi(s)=2\cosh h\,\nu_h(\varphi)$, applied at
every visible site.
\end{proof}

\begin{proposition}\label{prop:rational}
For the block assignment and all $a,b\in\N$, $\tau$, $H$, $\beta\ge0$,
the sequence $B_{bn,an}(H_{\mathsf t})$ converges as $n\to\infty$.
\end{proposition}

\begin{proof}
Fix $L\ge4$. Let $\widetilde\mu$ be the image of $\mu_{L}$ under
$x\mapsto x/c$, $\widetilde\Psi(R):=\Psi(c^2R)$ and $\widetilde X_n(S):=
X_n(cS)$. Then $\widetilde\mu(\mathsf B_{\mathsf D})=1$, $\widetilde\Psi$
is a polynomial with $\widetilde\Psi(0)=0$, and $\widetilde X_n$ is a
centered Gaussian process with covariance $n\widetilde\Psi(R_n)$ on all
of $\R^{\mathsf D\times n}$. Since $X_n(S)-\frac n2\Psi(R_n(S,S))=
\widetilde X_n(S/c)-\frac n2\widetilde\Psi(R_n(S/c,S/c))$ and $t=\frac12$
gives $\sqrt{2t}=1$, \eqref{eq:cancel} yields
\[
\frac1{bn}\E\log Z^{[L]}=\frac1b\sum_{u\le b}\log(2\cosh H_{\tau(u)})+
\frac ab\log d_L+\frac1b\Phi_n(1/2),
\]
with $\Phi_n$ as in~\eqref{eq:CHI} for $(\widetilde\mu,\widetilde\Psi,
\widetilde X_n)$. By Theorem~\ref{thm:CHI} the limit $\ell_L$ of the
left-hand side exists. By Proposition~\ref{prop:cut} with
$\alpha_\pm=a/b$, $\ell_L\le\liminf_nB_{bn,an}\le\limsup_nB_{bn,an}\le
\ell_L+C_{\rm cut}L^{-2}$. Let $L\to\infty$. Theorem~\ref{thm:CHI} is
applied at fixed $L$; the exchange of limits relies only
on~\eqref{eq:cut}.
\end{proof}

\subsection{Comparison of fields, volumes and loads}\label{sec:comparison}

\begin{lemma}[Fields]\label{lem:fields}
For a permutation $\pi$ of $\{1,\dots,N\}$,
$B_{N,K}(\vec h\circ\pi)=B_{N,K}(\vec h)$. For $\vec h,\vec h'\in\R^N$,
pointwise in $G$,
$|f^{V_{\vec h}}_N(G)-f^{V_{\vec h'}}_N(G)|\le N^{-1}\sum_i|h_i-h'_i|$.
\end{lemma}

\begin{proof}
The rows of $G$ are i.i.d., so permuting them leaves the law of $G$
unchanged; relabelling $\sigma$ accordingly proves the first claim. The
exponents differ by at most $\sum_i|h_i-h_i'|$ uniformly in $\sigma$.
\end{proof}

\begin{lemma}[Loads]\label{lem:patterns}
Let $V:\Sigma_N\to\R$ and $0\le K_-\le K_+$, $K_+\ge1$, and let
$G_\pm$ be standard Gaussian $N\times K_\pm$ matrices. Then
$0\le\E f^V_N(G_+)-\E f^V_N(G_-)\le\frac\beta2\frac{K_+-K_-}{K_+}
\big[(1+\sqrt{K_+/N})^2+\frac1N\big]$. In particular, if $\alpha_-N\le
K_+\le\alpha_+N$, the difference is at most $c_{\rm load}(K_+-K_-)/N$
with $c_{\rm load}:=\beta[(1+\sqrt{\alpha_+})^2+1]/(2\alpha_-)$.
\end{lemma}

\begin{proof}
Realize $G_-$ as the first $K_-$ columns of $G_+$ and let
$\Delta:=\frac\beta{2N}\sum_{\mu>K_-}(g^\mu\cdot\sigma)^2\ge0$. With
$\omega_+$ the Gibbs measure for $G_+$, Jensen's inequality gives
$0\le\log Z_+-\log Z_-=-\log\omega_+(e^{-\Delta})\le\omega_+(\Delta)$.
By column exchangeability (Lemma~\ref{lem:rotation}),
$\E\omega_+(\Delta)=\frac{K_+-K_-}{K_+}\frac\beta{2N}\E\omega_+
\norm{G_+^T\sigma}^2\le\frac{K_+-K_-}{K_+}\frac\beta2\E\norm{G_+}_\op^2$;
divide by $N$ and apply Lemma~\ref{lem:norm}.
\end{proof}

\begin{lemma}[Volumes]\label{lem:padding}
Let $d_0\in\N$, $\alpha_+>0$. There is $C_{\rm pad}(\beta,\alpha_+,d_0)$
such that for $0\le d\le d_0$, $K\le\alpha_+N$, $N':=N+d$, $\vec h\in
\R^N$ and $\vec h'\in\R^d$,
\[
\Big|N'B_{N',K}(\vec h,\vec h')-NB_{N,K}(\vec h)-\sum_{l\le d}
\log(2\cosh h'_l)\Big|\le C_{\rm pad}(d+\sqrt{Nd}).
\]
Consequently $|B_{N',K}(\vec h,\vec h')-B_{N,K}(\vec h)|\le
C(d+\sqrt{Nd})/N+2d(\log2+\norm{\vec h}_\infty+\norm{\vec h'}_\infty+
C)/N$, with $C=C(\beta,\alpha_+,d_0)$.
\end{lemma}

\begin{proof}
Write the $N'\times K$ matrix as $\binom{G_1}{G_2}$, $G_1$ of size
$N\times K$, and let $w:=G_1^T\sigma$, $w':=G_2^T\tau$ for
$\sigma\in\Sigma_N$, $\tau\in\Sigma_d$. The exponents differ by
$-\frac{\beta d}{2NN'}\norm w^2+\frac\beta{N'}w\cdot w'+\frac\beta{2N'}
\norm{w'}^2$, which is bounded in absolute value, uniformly in
$(\sigma,\tau)$, by
\[
\Delta:=\frac{\beta d}{2N'}\norm{G_1}_\op^2+
\frac\beta{N'}\sqrt{Nd}\norm{G_1}_\op\norm{G_2}_\op+\frac{\beta d}{2N'}
\norm{G_2}_\op^2.
\]
The sum over $\tau$ of $e^{\sum_lh'_l\tau_l}$ is $\prod_l2\cosh h'_l$,
so the difference in the first claim is at most $\E\Delta$. Since $G_2$
has $K$ columns, $\E\norm{G_2}_\op^2\le3(d_0+\alpha_+)N$, and
$\E\norm{G_1}_\op^2\le3(1+\alpha_+)N$; Cauchy--Schwarz gives
$\E\Delta\le C_{\rm pad}(d+\sqrt{Nd})$. For the second claim write
$B_{N',K}-B_{N,K}=\frac{1}{N'}[N'B_{N',K}-NB_{N,K}-\sum_l\log(2\cosh
h'_l)]+\frac1{N'}[\sum_l\log(2\cosh h'_l)-dB_{N,K}]$ and use
$|B_{N,K}(\vec h)|\le\log2+\norm{\vec h}_\infty+\frac32\beta(1+\alpha_+)$,
by~\eqref{eq:gradient} and Lemma~\ref{lem:norm}. Only bounded fields
are used below.
\end{proof}

The preceding comparisons remove the divisibility restrictions in the
block construction and allow arbitrary limiting proportions of site types.

\begin{theorem}[Gaussian bulk with finitely many site types]
\label{thm:types}
Let $D\ge1$ and let $\pi=(\pi_1,\dots,\pi_D)$ satisfy $\pi_j\ge0$ and
$\sum_j\pi_j=1$. For each $N$ let $\mathsf t^N$ be a type assignment with
$N^{-1}\#\{i:\mathsf t^N_i=j\}\to \pi_j$ for every $j$, and let $K(N)/N\to\alpha
\in(0,\infty)$. Then for every $\beta\ge0$ and $H\in\R^D$ the limit
\[
\mathfrak b(\alpha,\beta;H,\pi):=\lim_{N\to\infty}B_{N,K(N)}(H_{\mathsf t^N})
\]
exists, is finite, and depends on $(\mathsf t^N)$ and $(K(N))$ only through
$\pi$ and $\alpha$. The convergence is uniform on compact subsets of
$(\beta,H)\in[0,\infty)\times\R^D$.
\end{theorem}

\begin{proof}
\emph{Step 1: block sequences.} Assume first that $\pi$ is rational,
and let $\rho=a/b$ be rational. Multiplying $a$ and $b$ by a common
factor we may assume $b\pi_j\in\{0,1,2,\dots\}$ for every $j$,
and we choose $\tau$ with $\#\tau^{-1}(j)=b\pi_j$. By
Proposition~\ref{prop:rational}, $\ell_\rho:=\lim_nB_{bn,an}(H_{\mathsf t})$
exists for the block assignment.

\emph{Step 2: load $\lfloor\rho N\rfloor$.} Given $N$, let
$n:=\lfloor N/b\rfloor$, $N':=bn$ and $d:=N-N'<b$. Let $n_j:=\#\{i\le
N:\mathsf t^N_i=j\}$, so that $n_j=N\pi_j+o(N)$, while the block assignment on
$N'$ sites has exactly $nb\pi_j=N\pi_j+O(1)$ sites of type $j$. By the first
part of Lemma~\ref{lem:fields} we may permute the sites so that, for
every $j$, $\min(n_j,nb\pi_j)$ of the first $N'$ sites carry type $j$ at
positions where the block assignment has type $j$; the remaining first
$N'$ sites, at most $m_N:=\sum_j|n_j-nb\pi_j|=o(N)$ of them, have the
wrong type. By the second part of Lemma~\ref{lem:fields}, retyping them
changes $B$ by at most $2\max_j|H_j|\,m_N/N\to0$. By
Lemma~\ref{lem:padding}, removing the last $d$ sites changes $B$ by
$O(N^{-1/2})$. Finally $|\lfloor\rho N\rfloor-an|\le\rho b+1$, and
Lemma~\ref{lem:patterns} changes the load at cost
$c_{\rm load}(\rho b+1)/N'$. Hence $B_{N,\lfloor\rho N\rfloor}(H_{\mathsf t^N})
\to\ell_\rho$; in particular $\ell_\rho$ depends only on $\rho$, $\pi$
and $H$.

\emph{Step 3: all loads.} Let $c_I$ be $c_{\rm load}$ with
$\alpha_-=\alpha/2$, $\alpha_+=2\alpha$, and let $\rho\in\Q\cap
[\frac34\alpha,\frac32\alpha]$. For large $N$, both $K(N)$ and
$\lfloor\rho N\rfloor$ lie in $[\frac\alpha2N,2\alpha N]$, and
Lemma~\ref{lem:patterns} gives
$|B_{N,K(N)}-B_{N,\lfloor\rho N\rfloor}|\le c_I|K(N)-\lfloor\rho N
\rfloor|/N$. Thus $\limsup_NB_{N,K(N)}-\liminf_NB_{N,K(N)}\le2c_I|\alpha
-\rho|$; let $\rho\to\alpha$. The limit is finite, since $|B_{N,K}(\vec h)|\le\log2+
\norm{\vec h}_\infty+\frac32\beta(1+K/N)$ by~\eqref{eq:gradient} and
Lemma~\ref{lem:norm}, and two sequences with the same $\alpha$ and
$\pi$ have the same limit by the same comparisons.

\emph{Step 4: arbitrary proportions.} Steps 1--3 prove the claim for
rational $\pi$. For general $\pi$, let $\pi'$ be a rational probability
vector. Changing the types of at most $\frac N2\norm{\pi-\pi'}_1+o(N)$
sites turns $\mathsf t^N$ into an assignment $\mathsf t'^N$ with
proportions converging to $\pi'$, and by Lemma~\ref{lem:fields} each
change costs at most $2\max_j|H_j|/N$. Hence
\[
\limsup_N\big|B_{N,K(N)}(H_{\mathsf t^N})-\mathfrak b(\alpha,\beta;H,
\pi')\big|\le\max_j|H_j|\,\norm{\pi-\pi'}_1 .
\]
For the extension, put $M_H:=\max_j|H_j|$. For rational
probability vectors $\pi'$ and $\pi''$, the same retyping estimate,
applied to balanced assignments and then passed to the limit, gives
\[
\big|\mathfrak b(\alpha,\beta;H,\pi')-
\mathfrak b(\alpha,\beta;H,\pi'')\big|
\le M_H\norm{\pi'-\pi''}_1.
\]
Consequently, for any sequence of rational probability vectors
$\pi^{(k)}\to\pi$, the sequence
$\mathfrak b(\alpha,\beta;H,\pi^{(k)})$ is Cauchy. Its limit is
independent of the approximating sequence, by the same inequality;
denote it by $\mathfrak b(\alpha,\beta;H,\pi)$. Given $\epsilon>0$,
choose a rational $\pi'$ such that both
$M_H\norm{\pi-\pi'}_1<\epsilon$ and
$|\mathfrak b(\alpha,\beta;H,\pi')-
\mathfrak b(\alpha,\beta;H,\pi)|<\epsilon$. The preceding limsup bound
then yields
\[
\limsup_N\big|B_{N,K(N)}(H_{\mathsf t^N})-
\mathfrak b(\alpha,\beta;H,\pi)\big|\le2\epsilon.
\]
Let $\epsilon\downarrow0$. This proves existence of the limit and its
independence of the sequence of type assignments. Approximating any
two probability vectors by rational ones also gives
\[
\big|\mathfrak b(\alpha,\beta;H,\pi)-
\mathfrak b(\alpha,\beta;H,\pi')\big|
\le M_H\norm{\pi-\pi'}_1
\]
for arbitrary $\pi$ and $\pi'$.

\emph{Step 5: uniformity.} By Lemma~\ref{lem:fields},
$|B_{N,K}(H_{\mathsf t})-B_{N,K}(H'_{\mathsf t})|\le\max_j|H_j-H'_j|$, and
$|\partial_\beta B_{N,K}|\le\E\norm G_\op^2/(2N)\le\frac32(1+K/N)$ by
Lemma~\ref{lem:elementary}. The functions $(\beta,H)\mapsto
B_{N,K(N)}(H_{\mathsf t^N})$ are therefore equi-Lipschitz on compact sets, and
pointwise convergence implies uniform convergence on compact sets.
The derivative bound is uniform down to $\beta=0$. At that endpoint,
\[
B_{N,K(N)}(H_{\mathsf t^N})\big|_{\beta=0}
=\sum_{j=1}^D\frac{\#\{i:\mathsf t^N_i=j\}}N\log(2\cosh H_j)
\longrightarrow\sum_{j=1}^D\pi_j\log(2\cosh H_j).
\]
Thus the finite nets may be chosen in the closed parameter space
$[0,\infty)\times\R^D$. The uniformity statement includes compact sets
meeting $\beta=0$, and, with $D=1$, gives the corresponding assertion of
Theorem~\ref{thm:main}.
\end{proof}

\begin{proof}[Proof of Theorem~\textup{\ref{thm:main}} for Gaussian
patterns]
Take $D=1$, $\pi_1=1$ and $H_1=h$ in Theorem~\ref{thm:types}.
\end{proof}

\section{Binary patterns: mixed models and the lower bound}
\label{sec:mixed}

The finite-type Gaussian theorem yields the limits of mixed models
with finitely many retained binary patterns. A truncated disorder
comparison then gives the lower bound for the binary pressure.

For a real random variable $Q$ under a probability measure $\omega$ we
write $\kappa_{k,\omega}(Q)$ for its $k$-th cumulant; in particular
\[
\kappa_{4,\omega}(Q)=\omega Q^4-4\omega Q^3\omega Q-3(\omega Q^2)^2
+12\omega Q^2(\omega Q)^2-6(\omega Q)^4,
\]
and H\"older's inequality implies
$|\kappa_{4,\omega}(Q)|\le26\,\omega|Q|^4$.

\subsection{Mixed-model limits}
For $r\ge1$, $u\in[-1,1]^r$, $\beta\ge0$ and $h\in\R$, consider
$D=2^r$ types indexed by $s\in\{-1,1\}^r$, with proportions
$\pi_s=2^{-r}$ and fields $H_s(u;\beta,h):=\beta u\cdot s+h$, and let
\[
\mathfrak b_r(u;\alpha,\beta,h):=\mathfrak b\big(\alpha,\beta;H(u;\beta,h),(2^{-r})_s\big)
\]
be the corresponding limit of Theorem~\ref{thm:types}.

\begin{proposition}\label{prop:mixedlimit}
Let $K(N)/N\to\alpha\in(0,\infty)$ and fix $r\ge1$. For all
sufficiently large $N$, $K(N)\ge r$, and the resulting tail of the
sequence $p^{(r)}_{N,K(N)}(\beta,h)$ converges, uniformly on compact
subsets of $(\beta,h)\in[0,\infty)\times\R$, to
\begin{equation}\label{eq:pr-variational}
p_r(\alpha,\beta,h)=\sup_{u\in[-1,1]^r}\Big\{\mathfrak b_r(u;\alpha,\beta,h)
-\frac\beta2\norm u^2\Big\},
\end{equation}
which depends on $(K(N))$ only through $\alpha$. Moreover
$\log(2\cosh h)\le p_r(\alpha,\beta,h)\le\log2+|h|+C\beta(1+\alpha)$,
with $C$ independent of $r$.
\end{proposition}

\begin{proof}
Write $K=K(N)$ and $\vec m(\sigma):=N^{-1}(\zeta^a\cdot\sigma)_{a\le r}
\in[-1,1]^r$, which takes at most $(N+1)^r$ values. For $u\in[-1,1]^r$
let
\[
\widehat{\mathfrak b}_N(u):=\frac1N\log\sum_\sigma\exp\Big\{\frac\beta{2N}
\norm{G'^T\sigma}^2+\sum_i\Big(\beta\sum_{a=1}^ru_a\zeta_i^a+h\Big)
\sigma_i\Big\}.
\]
Since $\frac\beta{2N}\sum_a(\zeta^a\cdot\sigma)^2=\frac{\beta N}2
\norm{\vec m(\sigma)}^2$, decomposing according to the values of
$\vec m$ gives, pointwise,
\begin{equation}\label{eq:finiteMattis}
0\le\frac1N\log Z^{(r)}_{N,K}(\beta,h)-\sup_{u\in[-1,1]^r}
\Big\{\widehat{\mathfrak b}_N(u)-\frac\beta2\norm u^2\Big\}\le\frac{r\log(N+1)}N.
\end{equation}
The lower bound uses $\frac12\norm{\vec m}^2\ge u\cdot\vec m-\frac12
\norm u^2$; the upper bound takes $u=\vec m$ in each constrained sum.

The rows $(\zeta_i^1,\dots,\zeta_i^r)\in\{-1,1\}^r$ define a random type
assignment with counts $n_s$; site $i$ carries the field
$H_{s}(u;\beta,h)$ of its type. Conditionally on $\zeta$,
$\E_{G'}\widehat{\mathfrak b}_N(u)=B_{N,K-r}(H(u)_{\mathsf t})$ in the notation
of~\eqref{eq:Bdef}, and by Lemma~\ref{lem:fields} it depends on $\zeta$
only through the counts $n_s$. Let $\bar{\mathsf t}^N$ be a deterministic type
assignment with counts $\bar n_s$, $|\bar n_s-N2^{-r}|\le1$. After a
permutation, the two assignments differ at $\frac12\sum_s|n_s-\bar n_s|$
sites, where the fields differ by at most $2\beta\norm u_1\le2\beta r$.
By Lemma~\ref{lem:fields},
\[
\E\sup_{u}\big|\E_{G'}\widehat{\mathfrak b}_N(u)-B_{N,K-r}(H(u)_{\bar{\mathsf t}^N})\big|
\le\frac{\beta r}N\sum_s\E|n_s-\bar n_s|=O_{r,\beta}(N^{-1/2}).
\]
Theorem~\ref{thm:types}, applied with $(K-r)/N\to\alpha$ and the compact
set of fields $\{H(u;\beta,h):u\in[-1,1]^r\}$, gives
$B_{N,K-r}(H(u)_{\bar{\mathsf t}^N})\to \mathfrak b_r(u)$ uniformly in $u$. Conditionally
on $\zeta$, \ref{T:Poincare} and~\eqref{eq:gradient} give
$\Var_{G'}\widehat{\mathfrak b}_N(u)\le\beta^2N^{-2}\E\norm{G'}_\op^2=O(N^{-1})$,
uniformly in $u$ and $\zeta$. Finally $\partial_{u_a}\widehat{\mathfrak b}_N=
\beta N^{-1}\sum_i\zeta_i^a\omega(\sigma_i)\in[-\beta,\beta]$, so
$\widehat{\mathfrak b}_N$ and its limit are $\beta$-Lipschitz for the
$\ell^1$-norm. Fix $\theta>0$ and choose a finite set
$\mathcal U_\theta\subset[-1,1]^r$ such that every $u$ is within
$\ell^1$-distance $\theta$ of some $v\in\mathcal U_\theta$. Its
cardinality $m_\theta$ is independent of $N$. Put
$X_N(v):=\widehat{\mathfrak b}_N(v)-
\E_{G'}\widehat{\mathfrak b}_N(v)$, where the second term is conditional
on $\zeta$. The conditional variance bound and Cauchy--Schwarz give
\[
\E\max_{v\in\mathcal U_\theta}|X_N(v)|
\le\Big(\sum_{v\in\mathcal U_\theta}\E|X_N(v)|^2\Big)^{1/2}
\le\sqrt{m_\theta C/N}.
\]
For $u$ and a corresponding net point $v$, the two Lipschitz bounds
give an error at most $2\beta\theta$. Therefore
\begin{align*}
\E\sup_u|\widehat{\mathfrak b}_N(u)-\mathfrak b_r(u)|
&\le2\beta\theta+\sqrt{m_\theta C/N}\\
&\quad+\E\sup_u\big|\E_{G'}\widehat{\mathfrak b}_N(u)
-\mathfrak b_r(u)\big|.
\end{align*}
The last term tends to zero by the comparison of type counts and the
uniform convergence for balanced types proved above. First taking
$N\to\infty$ at fixed $\theta$ gives a limsup at most
$2\beta\theta$; then $\theta\downarrow0$ proves
$\E\sup_u|\widehat{\mathfrak b}_N(u)-\mathfrak b_r(u)|\to0$.
By~\eqref{eq:finiteMattis},
$|p^{(r)}_{N,K}-\sup_u\{\mathfrak b_r(u)-\frac\beta2\norm u^2\}|\le r\log(N+1)/N+
\E\sup_u|\widehat{\mathfrak b}_N-\mathfrak b_r|\to0$, which
proves~\eqref{eq:pr-variational} at fixed $(\beta,h)$. The functions
$p^{(r)}_{N,K}$ are $1$-Lipschitz in $h$ and satisfy $|\partial_\beta
p^{(r)}_{N,K}|\le\E\norm{(\zeta,G')}_\op^2/(2N)\le C(1+K/N)$ by
Lemmas~\ref{lem:elementary} and~\ref{lem:subG}. In particular, on the
tail $N\ge N_r$, these functions have a common Lipschitz constant $L$
for the metric $|\beta-\beta'|+|h-h'|$, and their pointwise limit
inherits that bound. If $\mathcal C\subset[0,\infty)\times\R$ is compact
and $\mathcal T_\theta$ is a finite $\theta$-net in $\mathcal C$, then
\[
\sup_{(\beta,h)\in\mathcal C}|p^{(r)}_{N,K(N)}(\beta,h)
-p_r(\alpha,\beta,h)|
\le2L\theta+
\max_{(\beta,h)\in\mathcal T_\theta}
|p^{(r)}_{N,K(N)}(\beta,h)-p_r(\alpha,\beta,h)|.
\]
For each fixed net, the maximum tends to zero by pointwise convergence;
letting $\theta\downarrow0$ proves local uniformity. This includes
compact sets meeting $\beta=0$, because the derivative bound is uniform
there and $p^{(r)}_{N,K(N)}(0,h)=p_r(\alpha,0,h)=\log(2\cosh h)$.
The bounds follow from Lemma~\ref{lem:elementary}, since
the full mixed matrix has independent unit sub-Gaussian entries and
Lemma~\ref{lem:subG} is uniform in $r\le K$.
\end{proof}

\subsection{Lower comparison}

The mixed-model limits provide comparison values for the binary
pressure. We first establish a replacement estimate uniform in the
potential generated by the retained patterns.

\begin{lemma}\label{lem:lower}
Let $V:\Sigma_N\to\R$ be deterministic. Let $J$ be an $N\times K$ matrix with independent entries,
$n_R$ of them symmetric Rademacher and the others standard Gaussian, let
$G$ be a standard Gaussian $N\times K$ matrix, and let $K/N\in
[\alpha_-,\alpha_+]$. Then for every $L\ge4$,
\[
\E f^V_N(J)\ge\E f^V_N(G)-\frac{C_{\rm cut}}{L^2}-\frac{13}3\beta^2
\frac{n_RL^4}{N^3}.
\]
In particular, for $J=\varepsilon$ and $N\ge4^6$,
$\E f^V_N(\varepsilon)\ge\E f^V_N(G)-(C_{\rm cut}+\frac{13}3\beta^2
\alpha_+)N^{-1/3}$, uniformly in $V$. If $V$ is random and independent
of $J$ and $G$, both bounds hold conditionally on $V$, and hence for the
unconditional expectations whenever $\E\max_\sigma|V(\sigma)|<\infty$.
\end{lemma}

\begin{proof}
Let $f^{V,[L]}_N(J):=N^{-1}\log\sum_\sigma e^{V(\sigma)}\int_{[-L,L]^K}
e^{\sqrt{\beta/N}\sigma^TJz}\ga^{\otimes K}(dz)$, so that $f^V_N\ge
f^{V,[L]}_N$ pointwise by~\eqref{eq:HS}. As a function of a single
entry $x=J_i^\mu$, the $k$-th derivative of $f^{V,[L]}_N$ is
$N^{-1}(\beta/N)^{k/2}\kappa_{k,\omega}(\sigma_iz_\mu)$, with $\omega$
the truncated Gibbs measure on $\Sigma_N\times[-L,L]^K$; in particular the fourth derivative is bounded by
$26\beta^2L^4/N^3$. The first three moments of a symmetric Rademacher
and a standard Gaussian variable agree, and their fourth moments add up
to $1+3=4$. Replacing the $n_R$ Rademacher entries one at a time by
Gaussian ones and using a third-order Taylor expansion with fourth-order
remainder therefore gives
$|\E f^{V,[L]}_N(J)-\E f^{V,[L]}_N(G)|\le\frac{26\cdot4}{24}
\frac{\beta^2L^4}{N^3}n_R=\frac{13}3\beta^2\frac{n_RL^4}{N^3}$.
Proposition~\ref{prop:cut} gives
$\E f^{V,[L]}_N(G)\ge\E f^V_N(G)-C_{\rm cut}L^{-2}$. For the last claim
take $n_R=NK$ and $L=N^{1/6}$.
\end{proof}

For each fixed $r$, let $\zeta^1,\dots,\zeta^r$ be the first $r$
columns of $\varepsilon$, condition on them, and use
$V(\sigma)=\frac\beta{2N}\sum_{a=1}^r(\zeta^a\cdot\sigma)^2+hNm(\sigma)$
with the remaining $K-r$ columns as bulk. Since $\max_\sigma|V(\sigma)|
\le N(\beta r/2+|h|)$, the integrability condition of
Lemma~\ref{lem:lower} holds. Then $A^R_{N,K}(h)$ and
$p^{(r)}_{N,K}(h)$ are the binary- and Gaussian-bulk pressures with this
potential, and Lemma~\ref{lem:lower} and
Proposition~\ref{prop:mixedlimit} yield
\begin{equation}\label{eq:lower}
\begin{gathered}
\liminf_NA^R_{N,K(N)}(h)\ge p_r(\alpha,\beta,h)\quad\text{for every
}r\ge1,\quad\text{hence}\\
\liminf_NA^R_{N,K(N)}(h)\ge\mathcal A_*:=\sup_{r\ge1}p_r(\alpha,\beta,h).
\end{gathered}
\end{equation}
The estimate is independent of the size of the potential, which may be
extensive and depend on $N$, and requires no Gibbs-overlap condition.
With $V=hNm$ and no retained column, Lemma~\ref{lem:lower} also gives
$A^R_{N,K}(h)\ge A^G_{N,K}(h)-CN^{-1/3}$.

\begin{remark}[Monotonicity in $r$]\label{rem:monotone}
For $r'>r$ apply Lemma~\ref{lem:lower} with $V$ built from the first
$r$ binary columns of the mixed model with $r'$ binary columns, and with
the bulk consisting of $r'-r$ binary and $K-r'$ Gaussian columns,
compared with a purely Gaussian bulk of $K-r$ columns. Only $n_R=
(r'-r)N$ entries are replaced, so
\[
p^{(r')}_{N,K}(h)\ge p^{(r)}_{N,K}(h)-C_{\rm cut}L^{-2}
-\tfrac{13}{3}\beta^2(r'-r)\frac{L^4}{N^2},
\]
and $N\to\infty$, then $L\to\infty$, gives $p_{r'}\ge p_r$.
\end{remark}

\section{Binary patterns: radial cutoff and sector decomposition}
\label{sec:sectors}

For the upper bound, the retained patterns must be selected from the
hidden coordinates. A radial cutoff bounds their number, and a sector
decomposition permits comparison with the mixed models.

\subsection{A radial hidden-spin cutoff}\label{sec:radial}
Return to the hidden representation of all $K$ binary columns:
\[
Z^R_{N,K}(h)=\sum_\sigma\int\ga^{\otimes K}(dz)
\exp\Big\{\sqrt{\frac\beta N}\,\sigma^T\varepsilon z+hNm(\sigma)\Big\}.
\]
Write $Z^{R,\mathrm{rad}}$ for the restriction to $\norm z^2\le\varrho
N$.

\begin{lemma}\label{lem:radial}
Let $K\le\alpha_+N$. There is $\varrho<\infty$, depending only on
$\beta$ and $\alpha_+$, such that
\[
0\le A^R_{N,K}(h)-N^{-1}\E\log Z^{R,\mathrm{rad}}_{N,K}(h)\le
-\frac1N\log(1-e^{-N/2})+2e^{-N}\Big(\beta K+\frac{CK}N\Big),
\]
with $C=C(\beta,\alpha_+)$, uniformly in $h$.
\end{lemma}

\begin{proof}
Let $c_\star$ be as in Lemma~\ref{lem:subG}, so that the event
$\{\norm\varepsilon_\op\le c_\star\sqrt N\}$ has probability at least
$1-2e^{-N}$, and take $\sqrt\varrho>\sqrt\beta c_\star+\sqrt{\alpha_+}
+1$. Conditionally on $\sigma$, the hidden law is $N(a,I_K)$, with
$a=\sqrt{\beta/N}\,\varepsilon^T\sigma$ and $\norm a\le\sqrt\beta\,
\norm\varepsilon_\op$, and $Z^{R,\mathrm{rad}}/Z^R=\omega\big(\Prob_Y
(\norm{Y+a}^2\le\varrho N)\big)$ with $\omega$ the Ising Gibbs measure.
On the good event, Gaussian norm concentration gives, uniformly in
$\sigma$, $\Prob(\norm{Y+a}>\sqrt{\varrho N})\le\Prob(\norm Y-\E\norm Y>
\sqrt N)\le e^{-N/2}$, so the normalized loss is at most
$-N^{-1}\log(1-e^{-N/2})$.

On the exceptional event, choose $d:=(\varrho/(2\alpha_+))^{1/2}$, so
that $Kd^2\le\varrho N/2$. Integrating the shifted Gaussian density on
$[-d,d]^K$ and using $\norm{z-a}^2/2\le\norm a^2+\norm z^2$ gives
$\Prob(\norm{Y+a}\le\sqrt{\varrho N})\ge\exp\{-\norm a^2-CK\}$ with
$C=d^2+(\log(\sqrt{2\pi}/(2d)))_+$. Thus the normalized loss is at most
$(\beta\norm\varepsilon_\op^2+CK)/N\le\beta K+CK/N$, since
$\norm\varepsilon_\op^2\le\norm\varepsilon_\HS^2=NK$.
\end{proof}

\subsection{Finite sectors}
Fix $\varrho$ as in Lemma~\ref{lem:radial} and $\delta>0$. Column $1$
is always kept binary; this reference column only keeps the index of the
mixed model at least one. For $S\subseteq\{1,\ldots,K\}$ with $1\in S$
let
\[
\mathcal D_S:=\Big\{z\in\R^K:\ \norm z^2\le\varrho N,\
|z_\mu|>\delta\sqrt N\ (\mu\in S\setminus\{1\}),\
|z_\mu|\le\delta\sqrt N\ (\mu\notin S)\Big\}.
\]
These domains partition the radial ball, since each $z$ determines
$S=\{1\}\cup\{\mu\ne1:|z_\mu|>\delta\sqrt N\}$. Only sets with
\[
|S|\le M_\delta:=1+\lfloor\varrho/\delta^2\rfloor
\]
can have positive measure, because $(|S|-1)\delta^2N<\norm z^2\le\varrho
N$ on $\mathcal D_S$. Let $\mathcal S_N$ be the family of sets $S$ with
$\ga^{\otimes K}(\mathcal D_S)>0$; then
\begin{equation}\label{eq:sectors}
|\mathcal S_N|\le M_\delta(K+1)^{M_\delta},\qquad
\log|\mathcal S_N|=O_{\varrho,\delta,\alpha_+}(\log N).
\end{equation}
For a real $N\times K$ matrix $J$ and $S\in\mathcal S_N$ let
\[
\Phi_S(J):=\frac1N\log Z_S(J),\qquad Z_S(J):=\sum_\sigma
e^{hNm(\sigma)}\int_{\mathcal D_S}e^{\sqrt{\beta/N}\,\sigma^TJz}
\ga^{\otimes K}(dz),
\]
which is finite and smooth. The domains depend only on the integration
variables and on $\varrho,\delta,N$, not on the disorder; in particular
the disorder keeps its product law in every sector.

\begin{lemma}[Small-coordinate comparison]\label{lem:small}
Keep $J^\mu=\varepsilon^\mu$ for $\mu\in S$ and replace all other
columns by independent standard Gaussian columns $G_{S^c}$. Then
\begin{equation}\label{eq:small}
\big|\E\Phi_S(\varepsilon)-\E\Phi_S(\varepsilon_S,G_{S^c})\big|
\le q_\delta:=\frac{13}{3}\beta^2\varrho\delta^2e^{4\sqrt\beta\,\delta},
\end{equation}
uniformly in $N,K,S,h$.
\end{lemma}

\begin{proof}
For $\mu\notin S$ interpolate $J_i^\mu(t)=\sqrt t\,\varepsilon_i^\mu+
\sqrt{1-t}\,g_i^\mu$, and let $\omega_t$ be the corresponding Gibbs
measure on $\Sigma_N\times\mathcal D_S$. Set $Q=\sigma_iz_\mu$ and let
$f(x)=\omega_t^{(x)}(Q)$ as a function of the binary entry
$\varepsilon_i^\mu=x$, all other variables being fixed. For a
symmetric sign $s$,
\begin{equation}\label{eq:stein}
|\E_s[sf(s)]-\E_sf'(s)|\le\frac23\sup_{|x|\le1}|f'''(x)|;
\end{equation}
indeed, expand $f$ to second order and $f'$ to first order at zero: the
leading terms agree, and the remainders are bounded by
$\norm{f'''}_\infty/6$ and $\norm{f'''}_\infty/2$.

Now $f'''(x)=(\beta t/N)^{3/2}\kappa_{4,\omega_t^{(x)}}(Q)$, so
$|f'''(x)|\le26(\beta t/N)^{3/2}\omega_t^{(x)}(z_\mu^4)$. Changing
$x\in[-1,1]$ to either actual sign changes the exponent by at most
$2\sqrt{\beta t}\,\delta$, because $|z_\mu|\le\delta\sqrt N$ on the
whole sector; hence the normalized Gibbs densities differ by a factor at
most $e^{4\sqrt\beta\,\delta}$, and for both values $s\in\{-1,1\}$,
$\sup_{|x|\le1}\omega_t^{(x)}(z_\mu^4)\le e^{4\sqrt\beta\,\delta}\,
\omega_t^{(s)}(z_\mu^4)$.

Let $q_S(t)=\E\Phi_S(J(t))$. For $0<t<1$,
\[
q_S'(t)=\frac{\sqrt\beta}{2N^{3/2}}\sum_{i,\,\mu\notin S}
\Big\{\frac{\E\varepsilon_i^\mu\omega_t(Q)}{\sqrt t}
-\frac{\E g_i^\mu\omega_t(Q)}{\sqrt{1-t}}\Big\}.
\]
Gaussian integration by parts gives $\E g_i^\mu\omega_t(Q)=
\sqrt{\beta(1-t)/N}\,\E[\omega_t(Q^2)-\omega_t(Q)^2]$, while
$\E f'(\varepsilon_i^\mu)=\sqrt{\beta t/N}\,\E[\omega_t(Q^2)-
\omega_t(Q)^2]$. The variance terms cancel, and~\eqref{eq:stein} gives
\[
|q_S'(t)|\le\frac{26}{3}\frac{\beta^2t}{N^3}e^{4\sqrt\beta\,\delta}
\sum_{i=1}^N\sum_{\mu\notin S}\E\omega_t(z_\mu^4).
\]
All terms in this sum use the same interpolating Gibbs measure, and on
$\mathcal D_S$, $\sum_{\mu\notin S}z_\mu^4\le\delta^2N\sum_{\mu\notin S}
z_\mu^2\le\delta^2\varrho N^2$. Hence $|q_S'(t)|\le\frac{26}3\beta^2t
\varrho\delta^2e^{4\sqrt\beta\delta}$, and integration in $t$
gives~\eqref{eq:small}. The differentiations are justified as follows.
By~\eqref{eq:sector-Lip} below, $\Phi_S$ is $\sqrt{\beta\varrho/N}$-Lipschitz
in the Hilbert--Schmidt norm, so $q_S$ is continuous on $[0,1]$. On
$[t_0,t_1]\subset(0,1)$, $|\partial_t\Phi_S(J(t))|\le\sqrt{\beta\varrho/N}
\,(\norm\varepsilon_\HS+\norm G_\HS)/(2\sqrt{t(1-t)})$, an integrable
bound that permits differentiation under the expectation. After the
cancellation, $|q_S'|$ is bounded on $(0,1)$, so $q_S(1)-q_S(0)=\int_0^1
q_S'(t)\,dt$.
\end{proof}

Dropping the sector constraints at the Gaussian endpoint increases the
partition function. Since the binary columns are i.i.d., the Gaussian
columns are i.i.d.\ and $\ga^{\otimes K}$ is permutation invariant,
relabelling the columns identifies the unconstrained model with binary
columns indexed by $S$ with the mixed model~\eqref{eq:mixed} with
$r=|S|$. The uniform field is unchanged; the reference column is not a
field source. Therefore
\begin{equation}\label{eq:endpoint-upper}
\E\Phi_S(\varepsilon)\le p^{(|S|)}_{N,K}(h)+q_\delta .
\end{equation}

\begin{remark}[Why the replacement is collective]\label{rem:collective}
Under the joint Gibbs measure $\omega_{\rm joint}$ of $(\sigma,z)$ on
$\Sigma_N\times\R^K$, without truncation, conditionally on $\sigma$ one
has $z_\mu=\sqrt{\beta N}\,m_\mu+Y_\mu$ with $m_\mu:=N^{-1}\varepsilon^\mu
\cdot\sigma$ and $(Y_\mu)$ independent standard Gaussian. A large hidden
coordinate does not by itself imply a large overlap, but for fixed
$\beta,\delta>0$,
\[
\omega_{\rm joint}\Big(\exists\mu:\ |z_\mu|>\delta\sqrt N,\ |m_\mu|\le
\frac\delta{2\sqrt\beta}\Big)\le2Ke^{-\delta^2N/8},
\]
since the event forces $|Y_\mu|>\delta\sqrt N/2$ for some $\mu$. In this
sense the sector $S$ retains the patterns associated with potentially
macroscopic overlaps; their number is at most $M_\delta$, and the proof
uses neither their identity nor any statement about condensed
equilibrium states. For the columns outside $S$ only $|z_\mu|\le
\delta\sqrt N$ is known. Using this bound entry by entry, as in
Lemma~\ref{lem:lower}, would cost $\frac{13}3\beta^2\delta^4/N$ per
entry and $\frac{13}3\beta^2\delta^4K$ in total, which diverges. The proof
of Lemma~\ref{lem:small} keeps one interpolating Gibbs measure for all
replaced entries and sums the fourth moments before bounding them; the
radial cutoff then makes $\sum_{\mu\notin S}z_\mu^4$ of order
$\delta^2N^2$ instead of $\delta^4N^3$.
\end{remark}

\subsection{Sector selection}\label{sec:selection}
Every $\Phi_S$ is convex in $J\in\R^{N\times K}$, as a normalized
log-Laplace transform of linear functions of $J$, and the radial bound
gives the global Lipschitz estimate
\begin{equation}\label{eq:sector-Lip}
\norm{\nabla\Phi_S(J)}_\HS\le L_N^{\rm sec}:=\sqrt{\beta\varrho/N}.
\end{equation}
Indeed, $\partial\Phi_S/\partial J_i^\mu=N^{-1}\sqrt{\beta/N}\,
\omega(\sigma_iz_\mu)$, so $\norm{\nabla\Phi_S}_\HS^2\le\beta N^{-3}
\sum_{i,\mu}\omega(z_\mu^2)=\beta N^{-2}\omega\norm z^2\le\beta\varrho/N$.
Since $Z^{R,\mathrm{rad}}=\sum_{S\in\mathcal S_N}Z_S(\varepsilon)$,
Lemma~\ref{lem:convex} gives
\[
\frac1N\E\log Z^{R,\mathrm{rad}}\le\E\max_S\Phi_S(\varepsilon)+
\frac{\log|\mathcal S_N|}N\le\max_S\E\Phi_S(\varepsilon)
+2\sqrt2\,L^{\rm sec}_N\sqrt{\log(|\mathcal S_N|+1)}+
\frac{\log|\mathcal S_N|}{N}.
\]
Combining~\eqref{eq:endpoint-upper}, \eqref{eq:sectors} and
Lemma~\ref{lem:radial}, we obtain, for each fixed $\delta>0$,
\begin{equation}\label{eq:upper-finite}
A^R_{N,K}(h)\le\max_{1\le r\le\min(K,M_\delta)}p^{(r)}_{N,K}(h)
+q_\delta+O_{\varrho,\delta,\beta,\alpha_+}\big(\sqrt{\log N/N}\big),
\end{equation}
uniformly in $h$.

\section{Binary patterns: conclusion of the proof}\label{sec:closing}

\begin{proof}[Proof of Theorems~\textup{\ref{thm:main}} for binary
patterns and~\textup{\ref{thm:Rrep}}]
At $\beta=0$ all quantities equal $\log(2\cosh h)$; let $\beta>0$. Fix
$\alpha,\beta,h$, let $\alpha_+:=\sup_NK(N)/N$ and choose $\varrho$ by
Lemma~\ref{lem:radial}. For every fixed $\delta>0$, $M_\delta$ is
finite, and $K(N)\ge M_\delta$ for all large $N$ since $\alpha>0$.
Proposition~\ref{prop:mixedlimit}, applied to each of the finitely many
$r\le M_\delta$, and~\eqref{eq:upper-finite} imply
\[
\limsup_NA^R_{N,K(N)}(h)\le\max_{1\le r\le M_\delta}p_r(\alpha,\beta,h)
+q_\delta\le\mathcal A_*+q_\delta.
\]
Together with~\eqref{eq:lower} and $q_\delta\to0$ as $\delta\downarrow0$,
\[
\mathcal A_*\le\liminf_NA^R_{N,K(N)}(h)\le\limsup_NA^R_{N,K(N)}(h)\le
\mathcal A_*,
\]
which proves~\eqref{eq:main}. The order of limits is: fix $\varrho$ and
$\delta$, let $N\to\infty$, then $\delta\downarrow0$. Convergence need not
be uniform in $r$, and $M_\delta$ need not stay bounded as
$\delta\downarrow0$. Sequence independence follows from that of every
$p_r$ (Proposition~\ref{prop:mixedlimit}); monotonicity in $r$ is
Remark~\ref{rem:monotone}, and boundedness is in
Proposition~\ref{prop:mixedlimit}, so $\sup_rp_r=\lim_rp_r$. Finally,
$A^R_{N,K}$ is $1$-Lipschitz in $h$ and $|\partial_\beta A^R_{N,K}|\le
\E\norm\varepsilon_\op^2/(2N)\le C(1+\alpha_+)$ by
Lemmas~\ref{lem:elementary} and~\ref{lem:subG}. These bounds hold for
all $\beta\ge0$, with a common Lipschitz constant $L$ for
$|\beta-\beta'|+|h-h'|$; passing to the pointwise limit gives the same
bound for $A^R$. Let $\mathcal C\subset[0,\infty)\times\R$ be compact,
and choose a finite $\theta$-net $\mathcal T_\theta\subset\mathcal C$.
For every parameter pair, comparison with a net point gives
\[
\sup_{(\beta,h)\in\mathcal C}
|A^R_{N,K(N)}(\beta,h)-A^R(\alpha,\beta,h)|
\le2L\theta+
\max_{(\beta,h)\in\mathcal T_\theta}
|A^R_{N,K(N)}(\beta,h)-A^R(\alpha,\beta,h)|.
\]
At fixed $\theta$, the finite maximum tends to zero. Taking the limsup
in $N$ and then letting $\theta\downarrow0$ proves local uniformity.
At $\beta=0$, both pressures equal
$\log(2\cosh h)$, and the same Lipschitz bounds hold up to the boundary.
This proves the uniformity assertions of Theorems~\ref{thm:main}
and~\ref{thm:Rrep}.
\end{proof}

\section{Concentration and the almost-sure thermodynamic limit}
\label{sec:concentration}

The convergence of the quenched means has now been established for both
pattern laws. We control sample fluctuations to obtain almost-sure
convergence, including local uniformity in the parameters.

\begin{proposition}[Explicit concentration bounds]\label{prop:explicit-conc}
Let $\beta>0$, $\alpha_+>0$, $N,K\ge1$ with $K\le\alpha_+N$, and let
$V:\Sigma_N\to\R$ be deterministic.
\begin{enumerate}[label=\textup{(\alph*)}]
\item \emph{Gaussian patterns.} Let $\delta_N:=e^{-N/4}\big(12c_V^2+
120\beta^2(1+\alpha_+)^2+24\beta^2\alpha_+(1+\alpha_+)N\big)^{1/2}$.
For every $t>0$,
\begin{equation}\label{eq:expconc}
\Prob\big(|f^V_N(G)-\E f^V_N(G)|\ge t+\delta_N\big)
\le2\exp\Big\{-\frac{Nt^2}{8\beta^2(1+\sqrt{\alpha_+})^2}\Big\}
+e^{-N/2}.
\end{equation}
\item \emph{Rademacher patterns.} Let $c_\star:=\big(8\{1+(1+\alpha_+)
\log9\}\big)^{1/2}$, $L_N:=\beta c_\star/\sqrt N$ and
$d_N:=(4c_V+\beta K)e^{-N}$. For every $t>0$,
\begin{equation}\label{eq:explicit-concentration}
\Prob\big(|f^V_N(\varepsilon)-\E f^V_N(\varepsilon)|>t+d_N+8\sqrt\pi
L_N\big)\le2e^{-N}+4\exp\Big\{-\frac{Nt^2}{16\beta^2c_\star^2}\Big\}.
\end{equation}
\end{enumerate}
\end{proposition}

\begin{proof}[Proof of Proposition~\textup{\ref{prop:explicit-conc}(a)}]
Let $f:=f^V_N(G)$, $R:=2(\sqrt N+\sqrt K)$, $\Lambda:=\beta R/N$ and
$\mathcal B:=\{\norm G_\op\le R\}$. Since $\mathcal B$ is convex,
\eqref{eq:gradient} and the mean value theorem show that $f$ is
$\Lambda$-Lipschitz on $\mathcal B$ for the Hilbert--Schmidt norm. Let
$g$ be its McShane extension~\ref{T:McShane}. As $\Lambda^2\le4\beta^2
(1+\sqrt{\alpha_+})^2/N$, \ref{T:Lipschitz} gives $\Prob(|g-\E g|\ge t)
\le2\exp\{-Nt^2/(8\beta^2(1+\sqrt{\alpha_+})^2)\}$. By
Lemma~\ref{lem:norm}, $\Prob(\mathcal B^c)\le e^{-(\sqrt N+\sqrt K)^2/2}
\le e^{-N/2}$. Since $g(0)=f(0)$ and $|f(0)|\le c_V$, we have $|g|\le c_V+
\Lambda\norm G_\HS$, so, by~\eqref{eq:gradient},
\[
\E(f-g)^2\le12c_V^2+\frac{3\beta^2}{4N^2}\E\norm G_\op^4+3\Lambda^2\E
\norm G_\HS^2\le12c_V^2+120\beta^2(1+\alpha_+)^2+24\beta^2\alpha_+
(1+\alpha_+)N,
\]
using $\E\norm G_\HS^2=NK$ and $\Lambda^2\le8\beta^2(N+K)/N^2$. As $f=g$
on $\mathcal B$, Cauchy--Schwarz gives $|\E f-\E g|\le(\E(f-g)^2)^{1/2}
\Prob(\mathcal B^c)^{1/2}\le\delta_N$. Hence $\{|f-\E f|\ge t+\delta_N\}
\subset\mathcal B^c\cup\{|g-\E g|\ge t\}$, which
proves~\eqref{eq:expconc}.
\end{proof}

The bound has the form $\exp(-cNt^2)$ of the concentration theorem
in~\cite{GT02}. The term $e^{-N/2}$ and the exponentially small shift
$\delta_N$ account for the unbounded gradient of the statistical
pressure.

\begin{proof}[Proof of Proposition~\textup{\ref{prop:explicit-conc}(b)}]
Let $f:=f^V_N$ on $\R^{N\times K}$ and $\mathcal K_N:=\{J:\norm J_\op
\le c_\star\sqrt N\}$, a compact convex set on which~\eqref{eq:gradient}
bounds the gradient by $L_N=\beta c_\star/\sqrt N$. Define the convex
extension by supporting planes
\[
\widetilde f(J):=\sup_{W\in\mathcal K_N}\{f(W)+\langle\nabla f(W),J-W
\rangle_\HS\}.
\]
It is finite, convex and $L_N$-Lipschitz; convexity of $f$ implies
$\widetilde f\le f$ everywhere, with equality on $\mathcal K_N$. Since
$0\in\mathcal K_N$ and $\nabla f(0)=0$ (Lemma~\ref{lem:gradient}),
$\widetilde f\ge f(0)\ge\log2-N^{-1}\max|V|$. On the binary cube
$\norm\varepsilon_\op^2\le NK$, so $f(\varepsilon)\le c_V+\beta K/2$ by
\eqref{eq:gradient}, and $0\le f(\varepsilon)-\widetilde f(\varepsilon)
\le2c_V+\beta K/2$. By Lemma~\ref{lem:subG},
$\Prob(\varepsilon\notin\mathcal K_N)\le2e^{-N}$, hence
$0\le\E f(\varepsilon)-\E\widetilde f(\varepsilon)\le d_N$. If $m_{\widetilde f}$ is a
median of $\widetilde f(\varepsilon)$, Lemma~\ref{lem:two-sided} gives
$|\E\widetilde f-m_{\widetilde f}|\le8\sqrt\pi L_N$ and
\[
\Prob\big(|f-\E f|>t+d_N+8\sqrt\pi L_N\big)\le\Prob(\varepsilon\notin
\mathcal K_N)+\Prob(|\widetilde f-m_{\widetilde f}|>t)\le2e^{-N}+4e^{-t^2/(16L_N^2)},
\]
which is~\eqref{eq:explicit-concentration}.
\end{proof}

\begin{proof}[Proof of Theorem~\textup{\ref{thm:concentration}}]
(a) In the Gaussian case, $c_V\le\bar c$ gives $\delta_N\le e^{-N/4}
(12\bar c^2+120\beta^2(1+\alpha_+)^2+24\beta^2\alpha_+(1+\alpha_+)N)^{1/2}
\to0$; in the binary case $d_N+8\sqrt\pi L_N\le(4\bar c+\beta\alpha_+N)
e^{-N}+8\sqrt\pi\beta c_\star N^{-1/2}\to0$. For $N\ge N_0$ these shifts
are at most $t/2$, and Proposition~\ref{prop:explicit-conc} with $t/2$ in
place of $t$ bounds the probability by $2e^{-cN}+e^{-N/2}$, respectively
$2e^{-N}+4e^{-cN}$, with $c>0$ depending on $\beta,\alpha_+,t$. For
$N<N_0$ the probability is at most $1\le e^{cN_0}e^{-cN}$.

(b)
If $\beta=0$, then $\widehat A^\bul_{N,K}(h)=\log(2\cosh h)$ is
deterministic; the finitely many $N$ with $K(N)=0$ are also
deterministic. Let $\alpha_+:=\sup_NK(N)/N$. By~\eqref{eq:gradient},
\ref{T:Poincare} and Lemma~\ref{lem:norm} in the Gaussian case, and by
Lemmas~\ref{lem:cubevar} and~\ref{lem:subG} in the binary case,
$\Var f^V_N(J^\bul)\le2\beta^2N^{-2}\E\norm{J^\bul}_\op^2
\le C_{\alpha_+,\beta}/N$ for every $V$; with $V=hNm$ and
Theorem~\ref{thm:main} this gives $L^2$ convergence.

For each rational pair $(\beta,h)$ with $\beta>0$ and each $t>0$, the
right-hand sides of~\eqref{eq:expconc} and
\eqref{eq:explicit-concentration} with $V=hNm$ are summable in $N$, and
$\delta_N,d_N,L_N\to0$. By~\ref{T:BC} and Theorem~\ref{thm:main},
$\widehat A^\bul_{N,K(N)}(\beta,h)\to A^\bul(\alpha,\beta,h)$ almost
surely, simultaneously for all rational pairs. By Lemmas~\ref{lem:norm}
and~\ref{lem:subG} and~\ref{T:BC}, almost surely $\norm{J^\bul_{N,K(N)}}
_\op\le c\sqrt N$ for all large $N$, with $c:=\max\{2(1+\sqrt{\alpha_+}),
c_\star\}$. On this event, by Lemma~\ref{lem:elementary}, $|\partial_h
\widehat A^\bul_{N,K}|\le1$ and $0\le\partial_\beta\widehat A^\bul_{N,K}
\le c^2/2$ for all parameters and all large $N$, while the limits are
locally Lipschitz by the same bounds on the means. Approximation by
finite rational nets gives convergence locally uniformly on
$[0,\infty)\times\R$. No independence between volumes has been used.
Finally, on a compact $\mathcal C\subset[0,\infty)\times\R$,
$\sup_{\mathcal C}|\widehat A^\bul_{N,K}-A^\bul|\le C_{\mathcal C}(1+
\norm{J^\bul}_\op^2/N)$ by Lemma~\ref{lem:elementary}, a family bounded
in $L^2$; uniform integrability and almost-sure convergence give $L^1$
convergence.
\end{proof}

\section{Internal energy, entropy and ground states}\label{sec:energy}

Convexity transfers the pressure limits to internal energy and entropy
at differentiability points. The entropy bound then controls the
zero-temperature limit.

\begin{proof}[Proof of Theorem~\textup{\ref{thm:energy}}]
Write $\widehat A_N:=\widehat A^\bul_{N,K(N)}(\cdot,0)$, $A_N:=\E
\widehat A_N$, $A:=A^\bul(\alpha,\cdot)$, $u_N:=u^\bul_N$, etc., and
$C_N:=\norm{J^\bul}_\op^2/(2N)$, so that $\sup_N\E C_N^q<\infty$ for
every $q<\infty$ by Lemmas~\ref{lem:norm} and~\ref{lem:subG}.

(i) For fixed disorder, $\beta\mapsto\widehat A_N(\beta)$ is smooth and
convex with $\partial_\beta\widehat A_N=-u_N$ and $\partial_\beta^2
\widehat A_N=N^{-1}\Var_{\omega_\beta}(H_N)\ge0$
(Lemma~\ref{lem:elementary}). Since $|u_N|\le C_N$ is integrable,
dominated convergence gives $\partial_\beta A_N=-E_N$, continuous in
$\beta$. From $\log\omega_\beta=-\beta H_N-\log Z_N$,
\begin{equation}\label{eq:entropyN}
s_N=\widehat A_N+\beta u_N,\qquad0\le s_N\le\log2,
\end{equation}
the bounds because $Ns_N$ is the Shannon entropy of a probability
measure on $2^N$ points; taking expectations gives $A_N=-\beta E_N+S_N$.

(ii)--(iii) $A$ is convex as a pointwise limit of convex functions,
hence differentiable outside a countable set~\ref{T:convex}.
Lemma~\ref{lem:convexderiv} with $f_N=A_N$ gives the one-sided bounds in
(iii) for $E_N$ and $E_N(\beta)\to-\partial_\beta A(\beta)$ at
differentiability points; then $S_N=A_N+\beta E_N\to A+\beta E=:S$, and
the bounds for $S_N$ follow from those for $E_N$. On the event of
Theorem~\ref{thm:concentration}(b), $\widehat A_N\to A$ on $(0,\infty)$,
and Lemma~\ref{lem:convexderiv} with $f_N=\widehat A_N$ gives the same
statements for $u_N$ simultaneously at all $\beta>0$;
\eqref{eq:entropyN} transfers them to $s_N$. For $1\le p<\infty$, the bound
$|u_N|\le C_N$ makes $(|u_N|^p)$ uniformly integrable, so almost-sure
convergence implies $L^p$ convergence; $s_N$ is bounded.

(iv) $E=-\partial_\beta A$ is nonincreasing since $\partial_\beta A$ is
nondecreasing. By convexity, $-E_N(\beta)\ge\partial_\beta A_N(0)=
\frac1{2N^2}\E\,2^{-N}\sum_\sigma\norm{(J^\bul)^T\sigma}^2=\frac{K(N)}{2N}$,
since $2^{-N}\sum_\sigma(J^\mu\cdot\sigma)^2=\norm{J^\mu}^2$ has
expectation $N$; so $E\le-\alpha/2$. In the Gaussian case
$E_N\ge-\E\norm G_\op^2/(2N)$ gives $E^G\ge-\frac12(1+\sqrt\alpha)^2$ by
Lemma~\ref{lem:norm}. The bounds on $S$ follow from~\eqref{eq:entropyN}.

For the additional pattern the same proof applies, using the convexity of
$\beta\mapsto\widehat A^{\bul,M}_N(\beta)$ (a log-Laplace transform),
Theorem~\ref{thm:Mattis} and $|u^{\bul,M}_N|\le C_N+\frac12$.
\end{proof}

\begin{remark}[Endpoints and possible transitions]
At $\beta=0$, $u^\bul_N(0)=-\norm{J^\bul}_\HS^2/(2N^2)$ and
$s^\bul_N(0)=\log2$; their limits are $-\alpha/2$ (almost surely in the
Gaussian case, by the moment bound in Remark~\ref{rem:diag}) and
$\log2$. If $A^\bul(\alpha,\cdot)$ is not differentiable at some
$\beta>0$, the internal energy may oscillate within the interval of
Theorem~\ref{thm:energy}(iii); such points would correspond to
first-order phase transitions. Differentiability everywhere, or a
unique limiting energy at a transition, would require a separate proof.
\end{remark}

\begin{remark}[Dimensionless field versus physical field]
The uniform field $h$ in~\eqref{eq:Z} is not multiplied by $\beta$. At
fixed $h$, with $u_N^{\rm int}=N^{-1}\omega(H_N)$ and $\widehat m_N=
\omega(m)$, one has $s_N=\widehat A^\bul_N(\beta,h)+\beta u_N^{\rm int}
-h\widehat m_N$, so that $S=A^\bul+\beta E-h\,\partial_hA^\bul$ wherever
both partial derivatives exist. To keep the identity $S=A+\beta E$ in a
physical field $B$, set $h=\beta B$ and include the field in the
Hamiltonian, $H_N^B=H_N-BNm$. The statistical pressure $A^\bul_B(\alpha,
\beta):=A^\bul(\alpha,\beta,\beta B)$ is convex in $\beta$, and the proof
of Theorem~\ref{thm:energy} applies verbatim, using $|N^{-1}\omega(H_N^B)|
\le C_N+|B|$: $E_B=-\partial_\beta A^\bul_B$ and $S_B=A^\bul_B+\beta E_B$
at every differentiability point of $A^\bul_B$.
\end{remark}

\begin{corollary}[Ground states]\label{cor:ground}
Let $\bul\in\{G,R\}$ and define the maximal interaction gain
\[
\widehat g^\bul_N:=-\frac1N\min_{\sigma\in\Sigma_N}H_N(\sigma;J^\bul)
=\frac{1}{2N^2}\max_{\sigma\in\Sigma_N}\norm{(J^\bul)^T\sigma}^2,
\qquad g^\bul_N:=\E\widehat g^\bul_N,
\]
with $K=K(N)$, and let $C_N:=\norm{J^\bul}_\op^2/(2N)$. Then
$g^\bul(\alpha):=\lim_{\beta\to\infty}A^\bul(\alpha,
\beta)/\beta$ exists and is finite, $g^\bul_N\to g^\bul(\alpha)$, and
$\widehat g^\bul_N\to g^\bul(\alpha)$ almost surely and in $L^1$. The
ground-state energy density is $-g^\bul(\alpha)$, and $g^G(\alpha)\le
g^R(\alpha)$.
\end{corollary}

\begin{proof}
Fix $\alpha$ and write $A(\beta):=A^\bul(\alpha,\beta)$. We first
establish the zero-temperature limit, then use it to bound the
finite-volume ground-state gains. Since $A$ is convex and $A(0)=\log2$, the secant slopes
\[
q(\beta):=\frac{A(\beta)-\log2}{\beta},\qquad\beta>0,
\]
are nondecreasing. Indeed, for $0<s<t$, convexity gives
$A(s)\le(1-s/t)A(0)+(s/t)A(t)$, hence $q(s)\le q(t)$.
Lemma~\ref{lem:elementary} and the matrix moment bounds give
$0\le q(\beta)\le C(1+\alpha)$ for all $\beta>0$. Thus $q(\beta)$ has a
finite limit, denoted $g^\bul(\alpha)$, and
\[
\frac{A(\beta)}{\beta}=q(\beta)+\frac{\log2}{\beta}
\longrightarrow g^\bul(\alpha).
\]
This limit depends on the sequence $(K(N))$ only through $\alpha$,
because $A$ does.

For every $\beta>0$, the finite-volume entropy bound gives
\[
\beta\widehat g^\bul_N\le\widehat A^\bul_N(\beta)
\le\beta\widehat g^\bul_N+\log2.
\]
Taking expectations and then the lower and upper limits in $N$ yields
\[
\frac{A(\beta)-\log2}{\beta}
\le\liminf_Ng^\bul_N\le\limsup_Ng^\bul_N
\le\frac{A(\beta)}{\beta}.
\]
Both endpoints tend to $g^\bul(\alpha)$ as $\beta\to\infty$, proving
$g^\bul_N\to g^\bul(\alpha)$.

For almost-sure convergence, let $\Omega_0$ be the event of probability one
in Theorem~\ref{thm:concentration}(b). Fix a realization in $\Omega_0$.
For every positive integer $j$, apply the same finite-volume inequality
at inverse temperature $\beta=j$ and use sample-pressure convergence
to obtain
\[
\frac{A(j)-\log2}{j}
\le\liminf_N\widehat g^\bul_N
\le\limsup_N\widehat g^\bul_N
\le\frac{A(j)}j.
\]
These bounds hold simultaneously for all $j$ on $\Omega_0$.
Letting $j\to\infty$ forces both the liminf and the limsup to equal
$g^\bul(\alpha)$, proving almost-sure convergence. Finally,
$0\le\widehat g^\bul_N\le C_N$ and $\sup_N\E C_N^2<\infty$ give
uniform integrability, hence convergence in $L^1$. The comparison
$g^G\le g^R$ follows by dividing $A^G\le A^R$
(Corollary~\ref{cor:compare}) by $\beta$ and letting $\beta\to\infty$.
\end{proof}

\begin{proof}[Proof of Proposition~\textup{\ref{prop:properties}}]
(a) The properties in $h$ and $\beta$ follow from
Lemma~\ref{lem:elementary} by passing to the limit. Let $\xi$ be a new
pattern, independent of $J^\bul$ and with the same law as its columns,
and let $\omega_K$ be the Gibbs measure with $K$ patterns. Jensen's
inequality gives $\log Z^\bul_{N,K+1}-\log Z^\bul_{N,K}=\log\omega_K
(e^{\beta(\xi\cdot\sigma)^2/(2N)})\ge\frac\beta{2N}\omega_K((\xi\cdot
\sigma)^2)$, and $\E_\xi(\xi\cdot\sigma)^2=N$ for every $\sigma$. Hence
$A^\bul_{N,K+1}-A^\bul_{N,K}\ge\frac\beta{2N}$, so $K\mapsto
A^\bul_{N,K}-\frac{\beta K}{2N}$ is nondecreasing. For $\alpha_1\le
\alpha_2$ and $K_i(N):=\lfloor\alpha_iN\rfloor$, Theorem~\ref{thm:main}
gives monotonicity of $\alpha\mapsto A^\bul-\beta\alpha/2$. (b) follows from~\eqref{eq:meanbounds} and the
annealed bound in Lemma~\ref{lem:elementary}. (c) The upper bound
follows from~\eqref{eq:meanbounds}, and by Lemma~\ref{lem:patterns},
$A^G_{N,K_2}-A^G_{N,K_1}\le\frac\beta2\frac{K_2-K_1}{K_2}[(1+
\sqrt{K_2/N})^2+\frac1N]$, which converges to the stated bound.
\end{proof}

\section{The additional binary pattern}\label{sec:Mattis}

We finally use the uniform-field limit to treat an additional Mattis
interaction. Gauge invariance removes the dependence on the law of the
additional pattern, and a finite-magnetization decomposition reduces
the problem to a scalar maximization.

For a real $N\times K$ matrix $J$ let $\widehat A^{\rm CW}_N(J):=
f^{\beta Nm^2/2}_N(J)$, the statistical pressure with an additional
Curie--Weiss term; here $c_V=\log2+\beta/2$.

\begin{lemma}[Gauge invariance]\label{lem:gauge}
Let $\bul\in\{G,R\}$, $D_\eta:=\mathrm{diag}(\eta)$ and $\widetilde J:=
D_\eta J^\bul$. Then $\widehat A^{\bul,M}_{N,K}(\beta)=\widehat
A^{\rm CW}_N(\widetilde J)$, and $\widetilde J$ has the law of $J^\bul$
and is independent of $\eta$. In particular $A^{\bul,M}_{N,K}=\E\widehat
A^{\rm CW}_N(J^\bul)$ does not depend on the law of $\eta$.
\end{lemma}

\begin{proof}
Under the bijection $\sigma\mapsto D_\eta\sigma$ of $\Sigma_N$,
$(J^\bul)^T\sigma=\widetilde J^T(D_\eta\sigma)$ and $\eta\cdot\sigma=
Nm(D_\eta\sigma)$. Given $\eta$, the entries $\eta_i(J^\bul)_i^\mu$ are
independent with the same symmetric law as those of $J^\bul$, so the
conditional law of $\widetilde J$ does not depend on $\eta$.
\end{proof}

\begin{lemma}\label{lem:Mattis}
For every real $N\times K$ matrix $J$,
\[
0\le\widehat A^{\rm CW}_N(J)-\max_{|v|\le1}\Big\{f^{\beta vNm}_N(J)-
\frac{\beta v^2}2\Big\}\le\frac{\log(N+1)}N .
\]
\end{lemma}

\begin{proof}
For $|v|\le1$, $m^2/2\ge vm-v^2/2$ gives the lower bound. Conversely, for
each of the $N+1$ values $x$ of $m(\sigma)$,
$\sum_{m(\sigma)=x}e^{-\beta H_N+\beta Nx^2/2}\le e^{-\beta Nx^2/2}
\sum_\sigma e^{-\beta H_N+\beta xNm(\sigma)}$, since each term on the left
appears on the right multiplied by $e^{\beta xNm(\sigma)-\beta Nx^2}=1$
and the remaining terms are positive.
\end{proof}

\begin{proof}[Proof of Theorem~\textup{\ref{thm:Mattis}}]
If $\beta=0$ all quantities equal $\log2$; let $\beta>0$ and write
$K=K(N)$. Since $f^{\beta vNm}_N(J^\bul)=\widehat A^\bul_{N,K}(\beta,
\beta v)$ and $|\max\varphi_1-\max\varphi_2|\le\sup|\varphi_1-\varphi_2|$,
Lemmas~\ref{lem:gauge} and~\ref{lem:Mattis} give
\[
\Big|A^{\bul,M}_{N,K}-\max_{|v|\le1}\Big\{A^\bul(\alpha,\beta,\beta v)
-\frac{\beta v^2}2\Big\}\Big|\le\frac{\log(N+1)}N+\E\sup_{|h|\le\beta}
|\widehat A^\bul_{N,K}(\beta,h)-A^\bul(\alpha,\beta,h)|\to0
\]
by Theorem~\ref{thm:concentration}(b). For $|v|>1$ and $s=\sgn v$, the
$1$-Lipschitz property of $A^\bul$ in $h$ gives $A^\bul(\alpha,\beta,
\beta v)-\beta v^2/2\le A^\bul(\alpha,\beta,\beta s)-\beta/2-\beta(|v|-
1)^2/2$, so the supremum over $\R$ equals the maximum over $[-1,1]$.

For the sample statements, apply Lemmas~\ref{lem:gauge}
and~\ref{lem:Mattis} with $\widetilde J_N:=D_{\eta_N}J^\bul_{N,K(N)}$.
The sequence $(\widetilde J_N)$ has the marginal laws of $J^\bul$ at
every volume, so Theorem~\ref{thm:concentration}(b), which does not
require independence between volumes, applies to it: almost surely
$f^{hNm}_N(\widetilde J_N)\to A^\bul(\alpha,\beta,h)$ (with $\beta$ the
inverse temperature in $f$)
locally uniformly in $(\beta,h)$, which permits passage to the maximum
in Lemma~\ref{lem:Mattis} simultaneously on compact $\beta$-intervals.
For $L^2$ convergence, Lemma~\ref{lem:gradient} with $V=\beta Nm^2/2$,
together with~\ref{T:Poincare} or Lemma~\ref{lem:cubevar}, gives
$\Var\widehat A^{\rm CW}_N(J^\bul)\le2\beta^2N^{-2}\E\norm{J^\bul}_\op^2
=O(N^{-1})$, and Lemma~\ref{lem:gauge} identifies the law of $\widehat
A^{\rm CW}_N(\widetilde J_N)$ with that of $\widehat A^{\bul,M}_{N,K}$.

Finally, for binary patterns choose $\eta$ with independent symmetric
Rademacher entries. The model~\eqref{eq:ZM} is then the binary Hopfield
model with $K(N)+1$ independent patterns, and since $(K(N)+1)/N\to
\alpha$, Theorem~\ref{thm:main} gives $A^{R,M}(\alpha,\beta)=
A^R(\alpha,\beta)$. As the value $v=0$ in~\eqref{eq:Mattislimit} gives
$A^{R,M}\ge A^R$, the identity is equivalent to the stated inequality.
\end{proof}

\begin{proof}[Proof of Corollary~\textup{\ref{cor:compare}}]
Adding the nonnegative term $\frac\beta{2N}(\zeta^1\cdot\sigma)^2$ to
the exponent shows $p^{(1)}_{N,K}(h)\ge A^G_{N,K-1}(h)$, hence $A^G\le
p_1$ by Theorem~\ref{thm:main}; $p_1\le A^R$ is~\eqref{eq:lower}. At
$h=0$, $p^{(1)}_{N,K}$ is the quenched pressure of the Gaussian model
with $K-1$ patterns and the additional pattern $\eta=\zeta^1$, so
Theorem~\ref{thm:Mattis} with $(K(N)-1)/N\to\alpha$ gives $p_1(\alpha,
\beta,0)=A^{G,M}(\alpha,\beta)$. The remaining assertions follow
from~\eqref{eq:main} and~\eqref{eq:R-fixed-point}.
\end{proof}

\section{Discussion}\label{sec:discussion}

For mean-field spin glasses, the
Guerra--Toninelli interpolation~\cite{GT02} established existence of the
thermodynamic limit before the proofs of the Parisi
formula~\cite{Par80,Gue03,Tal06} and of ultrametricity~\cite{Pan13}. For the Hopfield model, the lack of joint convexity of the bipartite
covariance requires a different argument. For Gaussian and binary
patterns, at every positive load and finite inverse temperature, we have
proved
convergence and self-averaging of the statistical pressure (also in a
uniform field and with an additional binary pattern), convergence of the
ground-state energy, and convergence of the internal energy and entropy
at every differentiability point of the limiting pressure; all limits
depend on $(K(N))$ only through $\alpha$. We do not establish the phase
diagram of Amit, Gutfreund and Sompolinsky~\cite{AGS85b,AGS87}: the
retrieval thresholds, the location and order of transitions, and the
concentration of order parameters such as Mattis magnetizations and
overlaps are not addressed. The limiting pressure is not evaluated explicitly. For Gaussian patterns, the representation below expresses the
limit implicitly, through the Hamilton--Jacobi variational formula
of~\cite{CHI} and a limit in the truncation level, and an explicit
variational principle has been announced in~\cite{CCM}; we do not prove
that the two coincide. For binary patterns, no formula in terms of conventional order
parameters is obtained here; equation~\eqref{eq:main} expresses $A^R$
through the Gaussian-bulk limits $p_r$.

Combining the proof of Proposition~\ref{prop:rational} with the
identification of the limit in~\cite[Theorem~1.1]{CHI} gives an implicit
variational representation. For rational $\rho=a/b$ and every
$L\ge4$ the limit $\ell_L$ of the truncated statistical pressure of the
Gaussian model is
\[
\ell_L=\log(2\cosh h)+\frac ab\log d_L-\frac1b\min_{\nu\in
\mathcal P(\mathrm{supp}\,\widetilde\mu)}\big\{\mathrm{Ent}(\nu\,|\,
\widetilde\mu)+f^\nu(\tfrac12,0)\big\},
\]
where $f^\nu$ is the viscosity solution of a Hamilton--Jacobi equation on
the space of monotone paths with nonlinearity $\widetilde\Psi$ and a
$\nu$-dependent initial condition (the sign reflects that~\cite{CHI}
state the result for $-\Phi_n$). By Proposition~\ref{prop:cut},
$A^G(\rho,\beta,h)$ is the limit of these values as $L\to\infty$, with
error at most $C_{\rm cut}L^{-2}$ uniformly in $h$. The same applies to
the Gaussian-bulk limits $\mathfrak b(\rho,\beta;H,\pi)$, with
$\log(2\cosh h)$ replaced by $b^{-1}\sum_{u\le b}\log(2\cosh
H_{\tau(u)})$, and, through~\eqref{eq:pr-variational}, to every $p_r$;
irrational loads follow by the load continuity of
Lemma~\ref{lem:patterns}. We do not determine whether
this representation reduces to the replica-symmetric expressions
of~\cite{AGS87,BGG10,ABBFP} in the appropriate region, or to the
variational principle announced in~\cite{CCM}. The Hamilton--Jacobi
structure follows the approach introduced by Guerra, in which mean-field
spin glasses are studied through partial differential equations from
analytical mechanics~\cite{Gue01,GB09,BDT13,BDGM14}, developed for exact
computations by Mourrat and collaborators~\cite{Mou20,Mou21b,Mou22,Mou26}.

The comparison between pattern laws reflects the distinct roles of
Gaussianity in the proof. The Gaussian argument uses
Gaussianity in three places: the blocked interaction is a Gaussian
process, as required by Theorem~\ref{thm:CHI}; the concentration
estimates use~\ref{T:Poincare}--\ref{T:Lipschitz}; and the truncation
estimate uses rotational invariance. For binary patterns,
Theorem~\ref{thm:CHI} and rotational invariance are used only for the
Gaussian bulk of the mixed models, the binary disorder being handled by
finite-pattern extraction and a Lindeberg-type comparison, and Gaussian
concentration is replaced by convex concentration on the discrete cube.
At zero field, Corollary~\ref{cor:compare} shows $A^G\le A^{G,M}\le A^R$. Whether the
supremum in~\eqref{eq:main} reduces to a fixed number of retained
patterns, in particular to one, remains open; at zero field this is the
question whether $A^R=A^{G,M}$, i.e.\ whether binary non-condensed
patterns act on the pressure as Gaussian noise. Identification with an
explicit hard--soft functional as in~\cite{ABBFP} would require control
of the interpolation remainder and of the order of its extrema. These identifications are not required for the existence theorem. The
finite-pattern extraction argument may also apply to other bipartite
architectures.

\section*{Acknowledgements}

A.B.\ acknowledges financial support from Sapienza University of Rome,
project number RM12519999AB8CA9 (``Neural Networks and Learning Machines:
asymptotic behaviors on structured datasets'') and project number
RP1261A106890761 (``Neural Exponential Engine 4 Data Localization
\& Extraction''). A.B.\ is a member of GNFM-INdAM, and F.D.\ is a member
of GNAMPA-INdAM.

\section*{Disclaimer on the use of AI tools}

The proofs presented in this paper were generated through extensive
interaction with the AI models Claude (Opus~5.5) and ChatGPT (Sol~6.1),
and the authors do not claim credit for their derivation. The authors
formulated the problem, guided the research by posing questions, and
suggested various approaches to its solution. This write-up is made
available as a service to the research community and is not intended
for journal publication in its current form. All statements and proofs
have been checked by the authors, who take responsibility for the
content of the paper.

\end{document}